\documentclass[12pt]{scrartcl}

\usepackage{fixltx2e}
\usepackage{lmodern}
\usepackage[T1]{fontenc}
\usepackage[utf8]{inputenc}
\usepackage{booktabs}
\usepackage{siunitx}
\usepackage{mathpazo}
\usepackage[margin=0.7in]{geometry}
\usepackage{titling}
\predate{}
\postdate{}
\usepackage{authblk}
\usepackage{amssymb}
\usepackage{amsmath}
\usepackage{xcolor}
\usepackage{graphicx}
\usepackage{enumitem}
\usepackage{float}
\usepackage{amsthm}

\newtheorem{theorem}{Theorem}

\newtheorem{lemma}[theorem]{Lemma}
\newtheorem{definition}{Definition}

\addtokomafont{disposition}{\rmfamily}
\newunit{\point}{pt}
\newunit{\inch}{in}

\title{Compactification of Extensive Game Structures and Backward Dominance Procedure\thanks{
The author owes the editors and two anonymous reviewers a great debt of gratitude for their comments and advice. She thanks Pierpaolo Battigalli, Dmitry Kvasov, and Andr\'{e}s Perea for valuable discussions and encouragements. She gratefully acknowledges the support of Grant-in-Aids for Young Scientists (B) of JSPS No.17K13707 and Grant for Special Research Projects No.2019C-484 and No. 2020C-018 of Waseda University.  }}
\author{Shuige Liu\thanks{School of Political Science and Economics, Waseda University, Nishi-Shinjuku 1-6-1, Shinjuku-Ku, 169-8050, Tokyo, Japan. EPICENTER, Maastricht University,
6200 MD Maastricht, The Netherlands. (\textsf{shuige\_liu@aoni.waseda.jp}) }}
\date{}

\begin{document}

\maketitle
\begin{abstract}
\textbf{Abstract} We study the relationship between invariant transformations on extensive game structures and backward dominance procedure (BD), a generalization of the classical backward induction introduced in Perea \cite{pe12}. We show that behavioral equivalence with unambiguous orderings of information sets, a critical property that guarantees BD's applicability, can be characterized by the classical Coalescing and a modified Interchange/Simultanizing in Battigalli et al. \cite{blm20}. We also give conditions on transformations that improve BD's efficiency. In addition, we discuss the relationship between transformations and Bonanno \cite{bo14}'s generalized backward induction.
\medskip

\end{abstract}

\section{Introduction}\label{sec:int}

A decision-maker in a dynamic situation is concerned at each move with what happened before and what may happen in the future. Hence the chronological order of plays, formulated by the arrangement of information sets, is a critical issue in dynamic epistemic game theory, a prosperous field which studies assumptions about strategic reasoning and their influence on behaviors in extensive games. See Battigalli and Bonnano \cite{bb99}, Perea \cite{pe12}, and Dekel and Siniscalchi \cite{ds15} for surveys of the field.

Here arises the problem of multiple representations. A dynamic decision-making situation can be represented by a class of extensive game \emph{structures}. Those structures describe the same information status, that is, a player's knowledge at each of her information set about another information set. They can be simplified to the same reduced form and transformed into each other through invariant operations. Their difference pertains to the orderings of information sets. See Battigalli et al. \cite{blm20}. The literature on \emph{game} equivalence, which is relevant to this issue, will be discussed in Section 1.1.

Some authors, like Kohlberg and Mertens \cite{km86}, claim that the multiplicity is not an essential obstacle. However, in dynamic epistemic game theory, for concepts related to backward induction, like Penta \cite{pe09}'s backwards rationalizability procedure, Perea \cite{pe14}'s belief in the opponents' future rationality, and Bonnano \cite{bo14}'s forward belief of rationality, the optimal strategies change with orderings of information sets. Further, some invariant transformation may mess up the ordering of information sets, making the algorithms impossible to be applied.

This sensitivity to orderings of information sets is due to the asymmetricity toward past and future. A player with a backward induction-related reasoning structure tends to ignore the past at each move and to believe that all her opponents will be rational henceforth. Therefore,  shifting-up an information set which was the future to the past changes her reasoning structure and may alter her optimal choices. In general, since the inconsistency between the past and the future rationality makes discrimination on the ground of chronological order unavoidable (see Reny \cite{re92}, \cite{re93}, Perea \cite{pe06}, \cite{pe07}), the multiplicity of orderings of information sets matters in epistemic game theory.

This paper takes Perea \cite{pe14}'s \emph{backward dominance procedure} (BD) as a representative and explores its relationship with invariant transformations on extensive game structures. BD is an algorithm that characterizes strategies that can be rationally chosen under common belief in future rationality if common belief in Bayesian updating is not imposed. For each information set, BD starts with the strategy profiles that reaches it; at each stage, the strategies that are strictly dominated at the information set \emph{and} at those simultaneous with or following it are eliminated. Both BD and Penta \cite{pe09}'s backwards rationalizability procedure generalize the classical backward induction in extensive games with imperfect information, while the latter eliminates strategies \emph{and} conditional beliefs.
BD is also related to Shimoji and Watson \cite{sw98}'s iterated conditional dominance procedure, while the elimination in the latter is symmetric toward past and future.\footnote {This feature makes the epistemic game theoretical concepts characterized by it, like extensive form rationality (Pearce \cite{pe84}, Battigalli \cite{ba97}) and common strong belief in rationality  (Battigalli and Siniscalchi \cite{bs02}), more robust to different orderings of information sets.} Bonnano \cite{bo14}'s generalized backward induction is another relevant concept and will be discussed later.

We do not follow the classical paradigm to characterize games having the same strategies surviving BD in terms of transformations since, as argued above, the latter alter the ordering of information sets and makes the survival invariance impossible. Instead, we raise two important conditions related to BD and investigate what invariant transformations satisfy them.

The first condition is called \emph{unambiguity of the ordering of information sets} (UO). As defined in Perea \cite{pe14}, it holds iff no information set is both before and after another one.\footnote{Following the tradition of Paul Halmos and Kelley \cite{ke55}, we use ``iff''  in definitions. The full expression ``if and only if'' is reserved for the logical connective indicating necessary and sufficient conditions in statements.} In other words, UO requires that the past does not get entangled with the future at each information set. If UO is violated, the elimination order between the two entangled information sets cannot be defined and BD fails to work.

\begin{figure}
\centering
  \includegraphics[width=0.8\columnwidth]{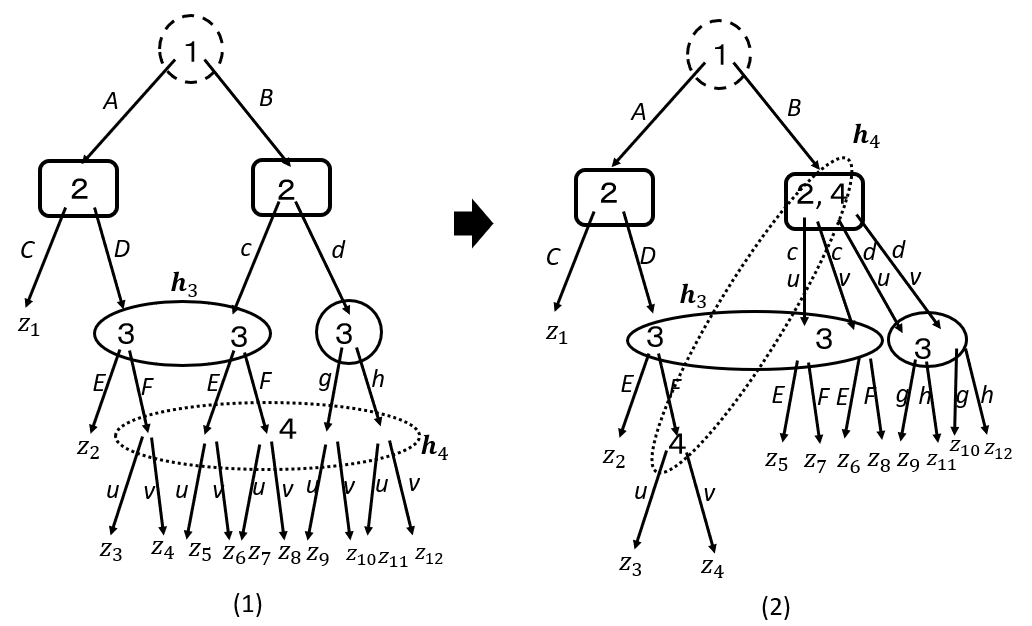}
  \caption{An interchanging/simultanizing destroys the the unambiguous ordering of information sets}
  \label{fig:MES}
\end{figure}

Our first question is what invariant transformations preserve UO. Since BD eliminates behaviorally equivalent strategies, or ``plans of actions'' \`a la Rubinstein \cite{ru91}, preserving behavioral equivalence is needed. As shown in Battigalli et al. \cite{blm20}, two transformations called \emph{Coalescing} and \emph{Interchange/Simultanizing} (IS) characterize behavioral equivalence. Roughly speaking, a Coalescing shifts up and combines one information set with another if both belong to the same player and the former is an extension of the latter; an IS synchronizes some histories in an information set with their common predecessor. Originated in Thompson \cite{to52}, both are classical invariant transformations compactifying extensive games.\footnote{We use the word \emph{compactification} instead of ``reduction'' adopted in the literature (e.g., Battigalli et al. \cite{blm20}) because we want to reserve the latter to the reduced normal forms/static structures.} However, an IS may destroy UO. For example, by applying IS on the structure in Figure \ref{fig:MES} (1), we obtained the structure in (2) where UO is violated since $\mathbf{h}_{4}$ is both before and after $\mathbf{h}_{3}$.

To preserve UO, we modify IS by prohibiting any ``partial crossing'' synchronization, called a \emph{non-crossing} IS. Theorem 1 shows that Coalescing and non-crossing IS characterize behavioral equivalence with UO. In other words, two extensive game structures \emph{satisfying UO} are behaviorally equivalent if and only they can be transformed into each other through a sequence of Coalescings, non-crossing ISs, and their inverses up to isomorphisms. This result provides a base for the applicability of BD in a behavioral equivalence class. It is a refinement of Battigalli et al. \cite{blm20}'s characterization of behavioral equivalence by Coalescing and IS.
\medskip

Our second question is when a compactification is \emph{monotonic}, that is, every strategy that can be eliminated in a structure can also be eliminated after a compactification. A monotonic compactification makes BD more efficient. Coalescing and non-crossing IS are not monotonic because both may remove the ``base'' for eliminating a strategy. Monotonicity requires a compactification to preserve simultaneity and to weakly preserve following, that is, after the compactification, two previously simultaneous information sets should be simultaneous, and one information following another still follows it or is simultaneous with it. Consequently, such a compactification should be ``collective'', that is, an information set shifts up if and only if all information sets related to it via simultaneity also shift up and their destinations are simultaneous. Also, the shifting-ups should only occur between neighbors because bypassing an in-between information set may cause some strategies reaching there no longer eliminatable.

In Section \ref{sec:cpa}, we formulate those requirements in a concept called \emph{complete immediate compactification opportunities} (complete ICO). Theorem 2 shows the transformation on a complete ICO is monotonic. We also proclaim that, though complete ICO does not characterize monotonicity, violating any of its conditions destroys monotonicity in some cases. In this sense, complete ICO can be regarded as necessary for monotonicity in a weak sense.

A drawback of iterative transformations on complete ICOs is the order-dependence. We introduce the \emph{backward compactification} which folds up a structure from the leaves to the root. Among all compactification orders, the backward compactification makes it possible to eliminate strategy at more information sets . Hence it may be taken as a benchmark.

\smallskip

We also briefly discuss Bonanno \cite{bo14}'s generalized backward induction (GBI), an algorithm closely related to BD. It is applied on terminal histories, instead of on strategy profiles in BD, to characterize \emph{forward belief of rationality}, a generalization of backward induction in a doxastic model which does not base on the classical subjective counterfactuals. GBI is defined on von Neumann structures in which all histories in an information set have the same length. In Section \ref{sec:dis}, we characterize behavioral equivalence in von Neumann structures by an invariant transformation and consider the monotonicity for GBI.

\subsection{The literature on game equivalence}

There are several kinds of equivalence in extensive games. For example, strategic equivalence defined in von Neumann and Morgenstern \cite{vnm} (pp.245-248) focuses on payoff functions leading to the same solution. Krentel et al. \cite{kmq51}, Thompson \cite{to52}, and Dalkey \cite{da53} defines \emph{game equivalence} in terms of reduced normal form. Roughly speaking, given an extensive game, we group a player's strategies that yield the same consequence against all her opponents' strategy profiles (called \emph{behavioral equivalence}) into an equivalence class and reduce the extensive game to a static one (called the \emph{reduced normal form}). Two extensive games are \emph{equivalent} iff they have the same reduced normal form up to isomorphisms.

Thompson \cite{to52} provides four elementary transformations on extensive games, called \emph{Inflation/Deflation}, \emph{Addition of a superfluous move}, \emph{Coalescing of moves}, and \emph{Interchange of moves}. He shows that two finite extensive games are equivalent if and only if they can be transformed into each other by applying the transformations finitely many times. His result is extended in later papers. The most famous one is Elmes and Reny \cite{er94} who notice that Inflation/Deflation may destroy perfect recall and characterize equivalence with perfect recall by Coalescing, Interchange of moves, and a modified version of Addition of a superfluous move. Recently, logicians started to use game algebra to explore game equivalence. See, for example, Goranko \cite{go03} and van Benthem et al. \cite{vbbe19}.

However, payoff equivalence in a game does not imply behavioral equivalence in the game structure. This problem hinders game equivalence's application  in the fields where game structures play a crucial role, like dynamic epistemic game theory. Endeavor devoted to fill this gap starts from Bonanno \cite{bo92}, which shows that the equivalence in the sense of set-theoretic forms can be characterized by Interchange of moves alone. Battigalli et al. \cite{blm20}, which profoundly influences this paper, characterize behavioral equivalence on extensive game \emph{structures} by Coalescing and Interchange/Simultanizing, the latter an adaption of Interchange of moves suiting simultaneous moves.

Initially, Thompson \cite{to52}'s  elementary transformations are algorithmic issues. They are re-interpreted in Kohlberg and Mertens \cite{km86} in terms of strategic features. They add two transformations to Thompson \cite{to52}'s and argue that all strategically stable equilibria are invariant under the six transformations. Their approach is inherited and extended in later researches. In that vein, de Bruin \cite{de99} studies the (in)variance of many important solution concepts under Kohlberg and Mertens \cite{km86}'s six transformations. Hoshi and Isaac \cite{hi10} studies the relationship between invariant transformations and games with unawareness.

This paper shares Kohlberg and Mertens \cite{km86}'s concern of transformations and strategic features of extensive games. However, we disagree with their assertion that ``elementary transformations [...] are irrelevant for correct decision-makings'' (p.1011). Chronological order is a nature of dynamics. As shown theoretically and experimentally (e.g., Amershi et al. \cite{ass89}, Hammond \cite{ha08}, Weber et al. \cite{wck04}), the orderings of information sets influence people's decision-making. This paper investigates how it influences the strategies optimal to common belief of future rationality. We anticipate more researches explicitly considering the relationship between time and reasoning structures in dynamic situations.

\smallskip
The rest of the paper is organized as follows. Section \ref{sec:pre} defines extensive game structures, unambiguous orderings of information sets, and behavioral equivalence. Section \ref{sec:tinf} introduces the invariant transformations and characterizes behavioral equivalence with unambiguous orderings of information sets. Section \ref{sec:cpa} studies backward dominance procedure and monotonic compactifications. Section \ref{sec:dis} discusses the necessity of conditions in Section \ref{sec:cpa} and the relationship between transformations and Bonanno \cite{bo14}'s generalized backward induction.

\section{Preliminaries}\label{sec:pre}
\subsection{Extensive Game Structures with Simultaneous Moves}

We adopt the history-based definition of an extensive game structure with simultaneous moves.\footnote{This definition originates in Osborne and Rubinstein \cite{or94}, Chapter 6. For its applications in epistemic game theory, see Battigalli and Bonanno \cite{bb99}, Battigalli and Siniscalchi \cite{bs02}, Bonanno \cite{bo14}, Perea \cite{pe12} (Chapters 8 and 9), \cite{pe14}, Battigalli \cite{ba20} (Chapter 10), to name but a few. The history-based definition is in essence equivalent to Kuhn \cite{ku53}'s classical one. See Al\'{o}s-Ferrer and Rizberger \cite{ar05}, \cite{ar08}, and \cite{ar16} (Chapters 2 and 3), for a detailed and insightful exploration on the relationship between the order  (history-originated) and graph (tree-originated) formulations of extensive game structures.} The formulation and terminologies follow Battigalli et al. \cite{blm20}.

\smallskip

Fix an arbitrary non-empty finite set $X$. A \emph{partition} $\mathbf{P}$ of $X$ is a set of subsets of $X$ satisfying (i) $P \neq \emptyset$ for each $P \in \mathbf{P}$, (ii) for each $P, Q \in \mathbf{P}, P \cap Q \neq \emptyset$ implies $P = Q$, and (iii) $\bigcup_{P \in \mathbf{P}}P=X$. We use $X^{*}$ to denote the set of finite sequences with each term in $X$. For sequences $x,y \in X^{*}$, $x$ is called a \emph{prefix} of $y$, denoted by $x \preceq y$,  iff $x =\emptyset$ or $x=x_{1}...x_{m}$, $y=y_{1}...y_{n}$, $m \leq n$, and $x_{i} = y_{i}$ for $i = 1,...,m$. The asymmetric part of $\preceq$ is denoted by $\prec$, i.e., $x \prec y$ iff $x \preceq y$ and $y \npreceq x$. We call $x$ an \emph{immediate predecessor} of $y$ or $y$ is an \emph{immediate successor} of $x$ iff $x \prec y$ and there is no $z$ such that $x \prec z \prec y$.

\smallskip

An \emph{extensive game structure} (abbreviated as \emph{structure}) is a tuple $\langle I, \bar{H},(A_{i},\mathbf{H}_{i})_{i \in I} \rangle$, where
\begin{itemize}
\item $I \neq \emptyset$ is the set of players.

\item For each $i \in I$, $A_{i}$ is the set of potentially feasible actions of player $i$. We define $A=\bigcup_{\emptyset \neq J \subseteq I} \bigl( \prod_{i \in J} A_{i} \bigr)$.

\item $\bar{H}$ is a finite subset of $A^{*}$, called the set of \emph{histories}. $(\bar{H},\preceq)$ is a finite tree, i.e., $\bar{H}$ contains the null history $\emptyset$  (called the \emph{root}) and $x \in \bar{H}$ implies that all prefixes of $x$ is in $\bar{H}$. A history $z$ is a \emph{terminal} iff $za \notin \bar{H}$ for all $a \in A$. We use $Z$ to denote the set of all terminals. Each $h \in H:= \bar{H} \setminus Z$ is called a \emph{non-terminal} history.

$H$ satisfies the following condition: for each non-terminal history $h \in H$, there is some $I(h) \subseteq I$ such that for each $a,a' \in A$ with $ha, ha' \in \bar{H}$, $a, a' \in \prod_{i \in I(h)} A_{i}$. The correspondence $I(\cdot): H \twoheadrightarrow I$ is called the \emph{active player correspondence}. 

For each $i \in I$, we define $H_{i}=\{h \in H: i \in I(h)\}$, i.e., the set of histories where player $i$ is active. We assume that $H_{i} \neq \emptyset$ for each $i \in I$, i.e., there is no idle player.

\item For each $i \in I$, $\mathbf{H}_{i}$ is a partition of $H_{i}$, called the \emph{information partition} of player $i$. It is an equivalence class describing  player $i$'s knowledge about what have occurred whenever it is her turn to move.

For each non-terminal history $h \in H$ and $i \in I(h)$,  we define $F_{i}(h)=$ Proj$_{A_{i}}\big\{a \in \prod_{i \in I(h)} A_{i}: ha \in \bar{H} \big\}$, i.e., the set of feasible actions for active player $i$ at $h$. We need the following assumptions:

\begin{itemize}
\item For each $h \in H$ and $i \in I(h)$, $|F_{i}(h)|\geq 2$. The motivation is that a player has no choice with only one feasible action and cannot be called active.\footnote{A tree satisfying this condition is called a \emph{decision tree} in the literature. It plays an important role in the general extensive game structures. See Al\'{o}s-Ferrer and Ritzberger \cite{ar05}, \cite{ar16} for a detailed discussion.}

\item For each $a \in A$ and $h \in H$, $ha \in \bar{H}$ if and only if $a \in \prod_{i \in I(h)}F_{i}(h)$. In words, at each history, every active player's decision-making is independent.

\item For each player $i \in I$, $F_{i}$ is $\mathbf{H}_{i}$-measurable, i.e., for each $\mathbf{h}_{i} \in \mathbf{H}_{i}$ and each $h,h^{\prime} \in \mathbf{h}_{i}$, $F_{i}(h)=F_{i}(h^{\prime})$. This classical requirements means that a player cannot differentiate histories in an information set by the available actions. Based on it, for each $\mathbf{h}_{i} \in \mathbf{H}_{i}$, we define $F_{i}(\mathbf{h}_{i})=F_{i}(h_{i})$ for some $h_{i} \in \mathbf{h}_{i}$. 

\item For each player $i$, $\mathbf{h}_{i}, \mathbf{h}_{i}^{\prime} \in \mathbf{H}_{i}$, $\mathbf{h}_{i} \neq \mathbf{h}_{i}^{\prime}$ implies $F_{i}(\mathbf{h}_{i}) \cap F_{i}(\mathbf{h}_{i}^{\prime}) = \emptyset$. This assumption is to simplify notations.
\end{itemize}

\end{itemize}

For example, the structure in Figure \ref{fig:RED} (1) is formulated as $I = \{1,2\}$, $\bar{H} = \{\emptyset, A,O,B$, $A(c,E),$ $A(c,F), A(d,E), A(d,F),$ $Oh, Oi, Bh, Bi\}$, $A_{1}=\{A,O,B,E,F\}$, $A_{2} = \{c,d,h,i\}$, $\mathbf{H}_{1} = \{\mathbf{h}_{11}, \mathbf{h}_{12}\}$, $\mathbf{H}_{2} = \{\mathbf{h}_{21}, \mathbf{h}_{22}\}$, where $\mathbf{h}_{11} = \{\emptyset\}$, $\mathbf{h}_{12} = \mathbf{h}_{21}=\{A\}$, and $\mathbf{h}_{22} = \{O, B\}$. 

\smallskip

For each $h \in \bar{H}$, we define $\lfloor h \rfloor = \{g \in \bar{H}: h \preceq g\}$ and  $Z(h) = Z \cap \lfloor h \rfloor$. In words, $\lfloor h \rfloor$ is the set of histories that can be reached from $h$ and $Z(h)$ is the terminals in $\lfloor h \rfloor$. This concept can be generalized to each subset $U \subset \bar{H}$. Namely, $\lfloor U \rfloor := \{g \in \bar{H}: h \preceq g$ for some $h \in U\}$ and $Z(U):=\cup_{h \in U}Z(h)$. For each $\mathbf{h}_{i} \in \mathbf{H}_{i}$ and $a_{i}^{\star} \in F_{i}(\mathbf{h}_{i})$, $\lfloor \mathbf{h}_{i}a_{i}^{\star} \rfloor := \cup_{h \in \mathbf{h}_{i}} \cup_{a_{-i} \in F_{-i}(h)}\lfloor h(a_{i}^{\star}, a_{-i}) \rfloor$ and $Z(\mathbf{h}_{i}a_{i}^{\star}):=Z \cap \lfloor \mathbf{h}_{i}a_{i}^{\star} \rfloor$.

Finally, we assume that a structure has perfect recall. For each $i \in I$ and $h \in H_{i}$, $X_{i} (h): = \{(\mathbf{h}_{i},a_{i}): h \in \lfloor \mathbf{h}_{i}a_{i} \rfloor\}$. In words, $X_{i} (h)$ is the pairs of player $i$'s information sets she encountered \emph{en route} for $h$ and the actions she took there. For example, in Figure \ref{fig:RED} (1), $X_{1}(A) = \{(\mathbf{h}_{11},A)\}$ and $X_{2}(O)$ is empty. A player $i \in I$ has \emph{perfect recall} iff $X_{i}(h) = X_{i}(h^{\prime})$ whenever $h, h^{\prime} \in H_{i}$ are in the same information set of $i$. A structure has perfect recall if each player has perfect recall. If a structure has perfect recall, then no history can go through an information set more than once.\footnote{Our definition follows Osborne and Rubinstein \cite{or94}, p. 203. See Piccione and Rubinstein \cite{pr97}, Al\'{o}s-Ferrer and Ritzberger \cite{ar16} (Chapter 6.4), \cite{ar17}, Battigalli et al. \cite{blm20},  Hillas and Kvasov \cite{hk20}, \cite{hk20b} for a detailed discussion.}

We use $\mathcal{G}$ to denote the set of extensive game structures.

\subsection{Unambiguous orderings of information sets}\label{sec:UO}

Fix an extensive game structure $G = \langle I, \bar{H},(A_{i},\mathbf{H}_{i})_{i \in I} \rangle$. We use $\mathbf{H}$ to denote the set of all information sets. Note that two information sets are distinct elements of $\mathbf{H}$ if and only if they contain different histories \emph{or} they belong to different players. For example, in Figure \ref{fig:RED} (1), $\mathbf{h}_{12}$ and $\mathbf{h}_{21}$, both containing only $A$, are distinct information sets since they belong to different players. In the following, sometimes we omit the subscript and use barely $\mathbf{h}$ as a representative element of $\mathbf{H}$. This omission only means that we focus on the set-theoretical aspect of the information set and the owner of it is inessential in the context there.

A history $h \in H$ is \emph{before} an information set $\mathbf{h} \in \mathbf{H}$ (or $\mathbf{h}$ \emph{follows} $h$), denoted by $h < \mathbf{h}$ (or $\mathbf{h} >h$), iff there is some $h^{\prime} \in \mathbf{h}$ such that $h \prec h^{\prime}$. Similarly, we can define $\mathbf{h} < h$. For two information sets $\mathbf{h}, \mathbf{h}^{\prime} \in \mathbf{H}$, we say that $\mathbf{h}$ is \emph{before} $\mathbf{h}^{\prime}$ or $\mathbf{h}^{\prime}$ \emph{follows} $\mathbf{h}$, denoted by $\mathbf{h} < \mathbf{h}^{\prime}$ or $\mathbf{h}^{\prime} > \mathbf{h}$, iff $h < \mathbf{h}^{\prime}$ for some $h \in \mathbf{h}$; $\mathbf{h}$ is \emph{simultaneous} with $\mathbf{h}^{\prime}$, denoted by $\mathbf{h} \sim \mathbf{h}^{\prime}$, iff $\mathbf{h} \cap \mathbf{h}^{\prime} \neq \emptyset$. We use $\mathbf{h} \lesssim \mathbf{h}^{\prime}$ as an abbreviation of ``$\mathbf{h} < \mathbf{h}^{\prime}$ or  $\mathbf{h} \sim \mathbf{h}^{\prime}$'', called $\mathbf{h}^{\prime}$ \emph{weakly follows} $\mathbf{h}$.

Since $G$ has perfect recall, for each $i \in I$, $\mathbf{H}_{i}$ is partially ordered, i.e., for each $\mathbf{h}_{i},\mathbf{g}_{i} \in \mathbf{H}_{i}$, at most one relation holds in $\mathbf{h}_{i} < \mathbf{g}_{i}$, $\mathbf{h}_{i} = \mathbf{g}_{i}$, and $\mathbf{h}_{i} > \mathbf{g}_{i}$. Yet those relations may not be mutually exclusive in $\mathbf{H}$. In Figure \ref{fig:KSM} (1), $\mathbf{h}_{2} <\mathbf{h}_{3}$, $\mathbf{h}_{2} \sim\mathbf{h}_{3}$, and $\mathbf{h}_{2} >\mathbf{h}_{3}$ all hold. A famous example where two information sets follow each other is the Kuhn-McKinsey-Shapley structure in Figure \ref{fig:KSM} (2) (see Kuhn \cite{ku53}).

\begin{figure}
\centering
  \includegraphics[width=0.8\columnwidth]{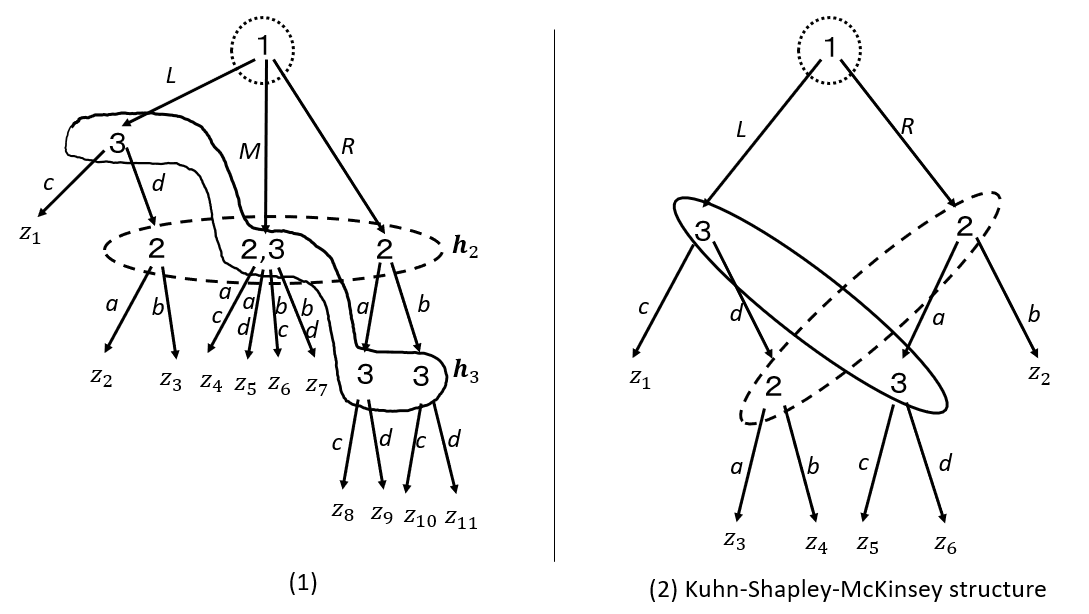}
  \caption{Two structures without unambiguous orderings of information sets}
  \label{fig:KSM}
\end{figure}

We say that $G$ has an \emph{unambiguous ordering} of the information sets iff for each $\mathbf{h}, \mathbf{h}^{\prime} \in \mathbf{H}$, $\mathbf{h} < \mathbf{h}^{\prime}$ implies $\mathbf{h} \ngtr \mathbf{h}^{\prime}$. In words, if $\mathbf{h}^{\prime}$ follows $\mathbf{h}$, then it cannot be followed by $\mathbf{h}$. We use UO as an abbreviation of the property. UO prohibits entangling the future with the past. The two structures in Figure \ref{fig:KSM} violate UO. Note that UO allows to mix up the present and the future (or the past). We use $\mathcal{G}_{UO}$ to denote the set of all extensive game structures satisfying UO.

A stronger notion is von Neumann structure. The the \emph{length} of each history $h \in H$, denoted by $\ell(h)$,  is defined inductively on $H$: $\ell(\emptyset)=0$ and $\ell(ha)=\ell(h)+1$ iff $h \in H$, $a \in \prod_{i \in I(h)}F_{i}(h)$, and $ha \in \bar{H}$. An information set $\mathbf{h}$ has the \emph{equal-length property} (EL) iff for each $h,h^{\prime} \in  \mathbf{h}$, $\ell(h)=\ell(h^{\prime})$. $G$ is called a \emph{von Neumann structure (vNM)} iff every information set of it satisfies EL.\footnote{This condition was required in von Neumann and Mogernstein \cite{vnm} (pp.60-66)'s formulation of games and was dropped in Kuhn \cite{ku53}.} Two information sets in a vNM have at most one relation in $<, \sim$, or $>$. Hence, a vNM satisfies UO.

 Bonanno \cite{bo14} shows that a structure satisfying UO can be transformed into a vNM by adding fictitious players (``the clocks'') having singleton information sets and only one action that can ``delay'' moves. However, this may cause problems in our context. A clock player has only one action, while we require that every active player at each history should have at least two actions. Adding a ``copy'' of the original action for the clock player leads to equalization of several terminals, which is not allowed in the framework of game structures here. Concerning those problems, we separate vNMs from structures satisfying UO.

\subsection{Behavioral equivalence}

For each $i \in I$, we define $\mathbf{H}_{i}^{o}$ to be the set of minimal information sets in $\mathbf{H}_{i}$ with respect to $<$. For example, in Figure \ref{fig:RED} (1), $\mathbf{H}_{1}^{o} = \{\mathbf{h}_{11}\}$ and $\mathbf{H}_{2}^{o} = \{\mathbf{h}_{21}, \mathbf{h}_{22}\}$. Since $G$ is finite and $\mathbf{H}_{i}$ is partially ordered, $\mathbf{H}_{i}^{o}$ is non-empty. Since $G$ has perfect recall, each $\mathbf{h}_{i} \in \mathbf{H}_{i} \setminus \mathbf{H}_{i}^{o}$ has a unique immediate predecessor $\mathbf{g}_{i}$ in $\mathbf{H}_{i}$ and a unique action in $F_{i}(\mathbf{g}_{i})$ leading to $\mathbf{h}_{i}$, denoted by $A_{i}(\mathbf{g}_{i}, \mathbf{h}_{i})$. For example, in Figure \ref{fig:RED} (1), $A_{1} (\mathbf{h}_{11},\mathbf{h}_{12}) = A$.

A \emph{strategy} of player $i \in I$ is a partial function $s_{i}: \mathbf{H}_{i} \hookrightarrow A_{i}$ satisfying the following conditions:\footnote{Recall that a partial function $f: X \hookrightarrow Y$ is a function whose domain is a subset of $X$.}
\begin{enumerate}[label=(\roman*)]
\item $s_{i}(\mathbf{h}_{i}) \in F_{i}(\mathbf{h}_{i})$ for each $\mathbf{h}_{i} $ in the domain of $s_{i}$,

\item $s_{i}$ is defined for each $\mathbf{h}_{i}^{o} \in \mathbf{H}_{i}^{o}$,

\item for each $\mathbf{h}_{i} \in \mathbf{H}_{i} \setminus \mathbf{H}_{i}^{o}$, $s_{i}$ is defined for $\mathbf{h}_{i}$ if and only if $s_{i}$ is defined for $\mathbf{h}_{i}$'s immediate predecessor $\mathbf{g}_{i}$ in $\mathbf{H}_{i}$ and $s_{i}(\mathbf{g}_{i}) = A_{i}(\mathbf{g}_{i},\mathbf{h}_{i})$. 
\end{enumerate}

The set of player $i$'s strategies is denoted by $\mathcal{S}_{i}$. We define $\mathcal{S} = \prod_{i \in I}$ as the set of \emph{strategy profiles}.  Each strategy profile uniquely determines a terminal history and we use $\zeta : \mathcal{S} \rightarrow Z$ to denote the mapping. The \emph{Z-reduced normal form} of $G$ is a tuple rn$_{Z}(G):=\langle I, (\mathcal{S}_{i})_{i \in I}, Z, \zeta \rangle$. Two structures $G,G^{\prime}$ are \emph{behaviorally equivalent} iff they have the same Z-reduced normal form up to isomorphisms. For example, the two structures in Figure \ref{fig:RED} are behaviorally equivalent.

\smallskip
In the classical definition, a strategy specifies an action for \emph{each} information set of the player even though some information sets would never be reached due to an action assigned before. Then, two strategies are called \emph{behaviorally equivalent} if they generate the same consequence against any strategy profiles of other players. Our definition packages the two steps. It corresponds to Rubinstein \cite{ru91}'s ``plan of action'' and is used in Perea \cite{pe12}. The term $Z$-reduced normal form is adopted from Battigalli  et al.\cite{blm20}; $Z$ here emphasizes the reduction is based on terminals instead of payoffs.

\begin{figure}
\centering
  \includegraphics[width=0.8\linewidth]{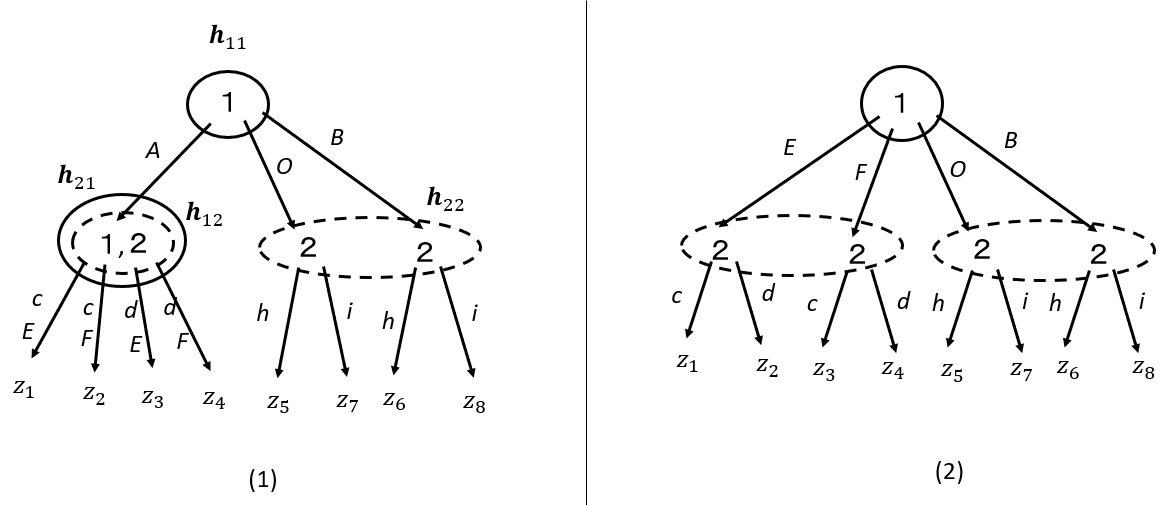}
  \caption{Two behaviorally equivalent game structures}
  \label{fig:RED}
\end{figure}

\section{Characterizing Unambiguous Orderings by Transformations}\label{sec:tinf}

In this section, we first introduce two classical transformations on extensive game structures, Coalescing and Interchanging/Simultanizing (IS). The latter is Battigalli et al \cite{blm20}'s adaption of Thompson \cite{to52}'s interchange of moves suiting simultaneous moves. We then define a modification of IS which preserves UO. Since those transformations make a structure more compact, we also call them \emph{compactifications}. Finally, we characterize behavioral equivalence with UO in terms of transformations.

Fix an extensive game structure $G =\langle I, \bar{H},(A_{i},\mathbf{H}_{i})_{i \in I} \rangle$. Consider a player $i \in I$ and her information set $\mathbf{h}_{i}$.  For $\mathbf{h}_{i}^{\prime} \in \mathbf{H}_{i}$, we say that $\mathbf{h}_{i}^{\prime}$ \emph{controls} $\mathbf{h}_{i}$, denoted by $\mathbf{h}_{i} \ll_{i} \mathbf{h}_{i}^{\prime}$, iff for some $a_{i}^{\ast} \in F_{i}(\mathbf{h}_{i})$, $Z(\mathbf{h}_{i}a_{i}^{\ast}) = Z(\mathbf{h}_{i}^{\prime})$. It can be seen that given $a_{i}^{\ast} \in F_{i}(\mathbf{h}_{i})$, there is at most one $\mathbf{h}_{i}^{\prime}$ controlling $\mathbf{h}_{i}$ and, in this case, $\mathbf{h}_{i}^{\prime}$ is the immediate successor of $\mathbf{h}_{i}$ in $\mathbf{H}_{i}$ with $A_{i}(\mathbf{h}_{i}$, $\mathbf{h}_{i}^{\prime}) = a_{i}^{\ast}$. For a non-terminal history $h \in H$ with $i \notin I(h)$ and some $\mathbf{d}_{i} \subseteq\mathbf{h}_{i}$, we say that $\mathbf{d}_{i}$ \emph{dictates} $h$, denoted by $h \lessdot_{i} \mathbf{d}_{i}$, iff $Z(h)=Z(\mathbf{d}_{i})$. Note that given $i \notin I(h)$, there is at most one $\mathbf{d}_{i}$ such that $h \lessdot_{i} \mathbf{d}_{i}$.\footnote{Dictation is called \emph{domination} in Battigalli et al \cite{blm20}. Here we use dictation because we want to reserve the term ``domination'' exclusively for the domination between strategies.}

\subsection{Coalescing}
Consider  $\mathbf{h}_{i}, \mathbf{h}_{i}^{\prime} \in \mathbf{H}_{i}$ such that $\mathbf{h}_{i} \ll_{i} \mathbf{h}_{i}^{\prime}$ and $A_{i}(\mathbf{h}_{i},\mathbf{h}_{i}^{\prime})=a_{i}^{\star}$. We say that $G$ has a \emph{Coalescing opportunity} $(\mathbf{h}_{i},\mathbf{h}_{i}^{\prime})$ and define $\gamma (G; \mathbf{h}_{i},\mathbf{h}_{i}^{\prime}) =\langle \tilde{I}, \tilde{\bar{H}},(\tilde{A}_{i},\tilde{\mathbf{H}}_{i})_{i \in \tilde{I}} \rangle$ as follows:

\begin{itemize}
\item $\tilde{I} = I$ and $\tilde{A}_{j}=A_{j}$ for each $j \in \tilde{I}$. 

\item $\tilde{\bar{H}}$ coincides with $\bar{H}$ at histories before $\mathbf{h}_{i}$ or unrelated to  $\mathbf{h}_{i}$ or $\mathbf{h}_{i}^{\prime}$. Formally, for each $g \in \bar{H}$, $g \in \tilde{\bar{H}}$ if and only if $g$ satisfies one condition in the following: \textbf{(i)} $g < \mathbf{h}_{i}$ or $g \in \mathbf{h}_{i}$, 
 \textbf{(ii)} $g$ and $\mathbf{h}_{i}$ are at different branches, i.e., for each $h \in  \mathbf{h}_{i}$ and its maximal common prefix $f$ with $g$, $f \prec h$ and $f \prec g$; \textbf{(iii)} $g$ follows $\mathbf{h}_{i}$ but $g$ and $\mathbf{h}_{i}^{\prime}$ are at different branches, i.e., for each  $h^{\prime} \in \mathbf{h}_{i}^{\prime}$, its maximal common prefix with $g$ is strictly before both $h^{\prime}$ and $g$.

For each $g$ with $g \succeq h(a_{i}^{\star}, a_{-i})$ for some $h \in  \mathbf{h}_{i}$ and $a_{-i} \in F_{-i}(h)$, $g$ need to be replaced in $\tilde{\bar{H}}$ due to the shifting-up of $\mathbf{h}_{i}^{\prime}$ to $\mathbf{h}_{i}$. Let $F_{i}(\mathbf{h}_{i}^{\prime})=\{c_{i1},...,c_{ik}\}$.  There are two cases:

\begin{itemize}
\item $g$ is between $\mathbf{h}_{i}$ and $\mathbf{h}_{i}^{\prime}$, i.e.,  $g = ha_{i}^{\star}a_{-i}b^{1}...b^{p}$  before or contained in $\mathbf{h}_{i}^{\prime}$. Then $g$ is replaced by $k$ replicas $\tilde{g}_{1}, ..., \tilde{g}_{k}$ in $\tilde{H}$: for each $t=1,...,k$, $\tilde{g}_{t} = h(c_{it}, a_{-i})b^{1}...b^{p}$;

\item $g>\mathbf{h}_{i}^{\prime}$, i.e., $g = h(a_{i}^{\star}, a_{-i})b^{1}...b^{p}(c_{it}, c_{-i})d^{1}...d^{q}$ for some $t \in \{1,...,k\}, p,q \in \mathbb{N}_{0}:=\{0,1,2,...\}$. Then $g$ is replaced in $\tilde{\bar{H}}$ by $\tilde{g} = h(c_{it}, a_{-i})b^{1}...b^{p}c_{-i}d^{1}...d^{q}$. Note that if at the prefix of $g$ in $\mathbf{h}_{i}^{\prime}$ only player $i$ is active, we directly connect the components before and after it, i.e., $\tilde{g} = h(c_{it}, a_{-i})b^{1}...b^{p}d^{1}...d^{q}$; especially, if $g$ is an immediate successor of some histories in $\mathbf{h}_{i}^{\prime}$, then $\tilde{g}$ coincides with some replica of its immediate predecessor.

\end{itemize}

\item Each $\tilde{\mathbf{H}}_{j}$ is modified accordingly. $\mathbf{h}_{i}^{\prime}$ disappears in $\tilde{\mathbf{H}}_{i}$. If an information set $\mathbf{g}_{j}$ contains some $g$ between $\mathbf{h}_{i}$ and $\mathbf{h}_{i}^{\prime}$, in $\tilde{\mathbf{H}}_{j}$ it should replace $g$ by $\tilde{g}_{1}, ..., \tilde{g}_{k}$; if $\mathbf{g}_{j}$ contains some $g$ following $\mathbf{h}_{i}^{\prime}$, then in $\tilde{\mathbf{H}}_{j}$ it should replace $g$ by $\tilde{g}$. Here, since all information sets except $\mathbf{h}_{i}^{\prime}$ still exist in $\gamma (G; \mathbf{h}_{i},\mathbf{h}_{i}^{\prime})$, we use the same symbol for each information set before and after the transformation when no confusion is caused.

\end{itemize}
For example, $(\mathbf{h}_{21},\mathbf{h}_{22})$ is a Coalescing opportunity in the structure in Figure \ref{fig:Coa1} (1). After coalescing $\mathbf{h}_{22}$ with $\mathbf{h}_{21}$,  we obtain $\gamma(G;\mathbf{h}_{21},\mathbf{h}_{22})$ in (2). Note that each history between $\mathbf{h}_{21}$ and $\mathbf{h}_{22}$, for example, $Rb$, is replaced by two replicas, $Rp$ and $Rq$. Histories following $\mathbf{h}_{22}$ replace $b$ by the corresponding actions in $F_{2}(\mathbf{h}_{22})$. For example, $Rb(q,H)v$ in $G$ is replaced by $RqHv$.

\begin{figure}
\centering
  \includegraphics[width=0.8\columnwidth]{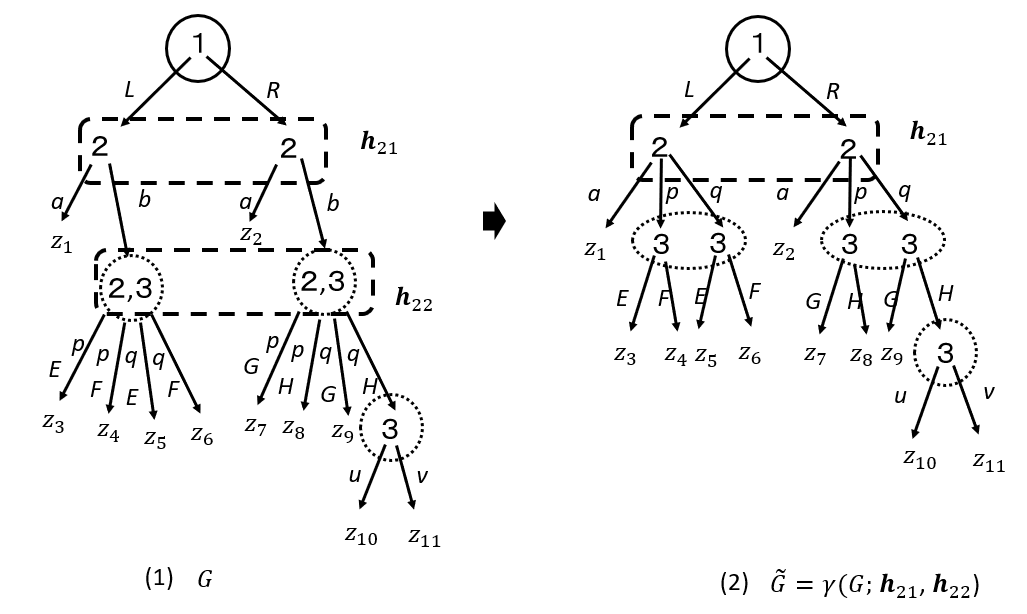}
  \caption{An example of Coalescing}
  \label{fig:Coa1}
\end{figure}

\smallskip
There is a natural correspondence, denoted by $\gamma: \bar{H} \twoheadrightarrow \tilde{\bar{H}}$, which indicates the change of each history through the transformation: for each $h \in \bar{H}$, if $h$ is before or unrelated to $\mathbf{h}_{i}$ or $\mathbf{h}_{i}^{\prime}$, $\gamma (h) = \{h\}$; if $h$  is between $\mathbf{h}_{i}$ and $\mathbf{h}_{i}^{\prime}$, then $\gamma (h) = \{\tilde{h}_{1},...,\tilde{h}_{k}\}$; if $h > \mathbf{h}_{i}^{\prime}$,  $\gamma (h) = \{\tilde{h}\}$. It can be seen that $|\gamma (h) | = |F_{i}(\mathbf{h}_{i}^{\prime})|$ if $h$ is between $\mathbf{h}_{i}$ and $\mathbf{h}_{i}^{\prime}$, and
$|\gamma (h) | = 1$ otherwise. In the following, when $\gamma (h)$ is a singleton, we sometimes do not differentiate $\gamma (h)$ from the unique history in it. For $h \in \mathbf{h}_{i}^{\prime}$, we define $\gamma (h;\mathbf{h}_{i}^{\prime})= \{g: g \in \mathbf{h}_{i} \text{ and } g \prec h\}$ to emphasize that \emph{for player} $i$, the history $h$ is now shifted up to its prefix in $\mathbf{h}_{i}$.

There is also a natural surjection $\gamma: \mathbf{H}\rightarrow \tilde{\mathbf{H}}$ such that $\gamma (\mathbf{h}_{i}^{\prime}) = \mathbf{h}_{i}$ and $\gamma (\mathbf{g}_{j}) = \mathbf{g}_{j}$ for every other information set $\mathbf{g}_{j}$.

Battigalli et al. \cite{blm20} show that a Coalescing preserves behavioral equivalence. The following statement shows that it also preserves UO.

\begin{lemma}
\textbf{(A Coalescing preserves UO)}. Let $(\mathbf{h}_{i},\mathbf{h}^{\prime}_{i})$ be a Coalescing opportunity of $G$. If $G$ satisfies UO, so does $\gamma (G; \mathbf{h}_{i},\mathbf{h}^{\prime}_{i})$.
\end{lemma}

\begin{proof}

Assume that $G$ satisfies UO and let $\tilde{G} = \gamma (G; \mathbf{h}_{i},\mathbf{h}^{\prime}_{i})$.  Since the relative positions between two information sets (except with $\mathbf{h}_{i}^{\prime}$) does not change, $\tilde{G}$ satisfies UO. 
\end{proof}

The inverse of $\gamma$ may break UO. For example, the structure in Figure \ref{fig:Coaib} (1) does not satisfy UO since both $\mathbf{h}_{22}>\mathbf{h}_{3}$ and $\mathbf{h}_{22}<\mathbf{h}_{3}$ hold. However, after coalescing $\mathbf{h}_{22}$ with $\mathbf{h}_{21}$, in Figure \ref{fig:Coaib} (2), UO is restored. We restrict the range of $\gamma^{-1}$ to the structures satisfying UO. Then the converse of Lemma 1 holds by definition.
\begin{figure}
  \includegraphics[width=\linewidth]{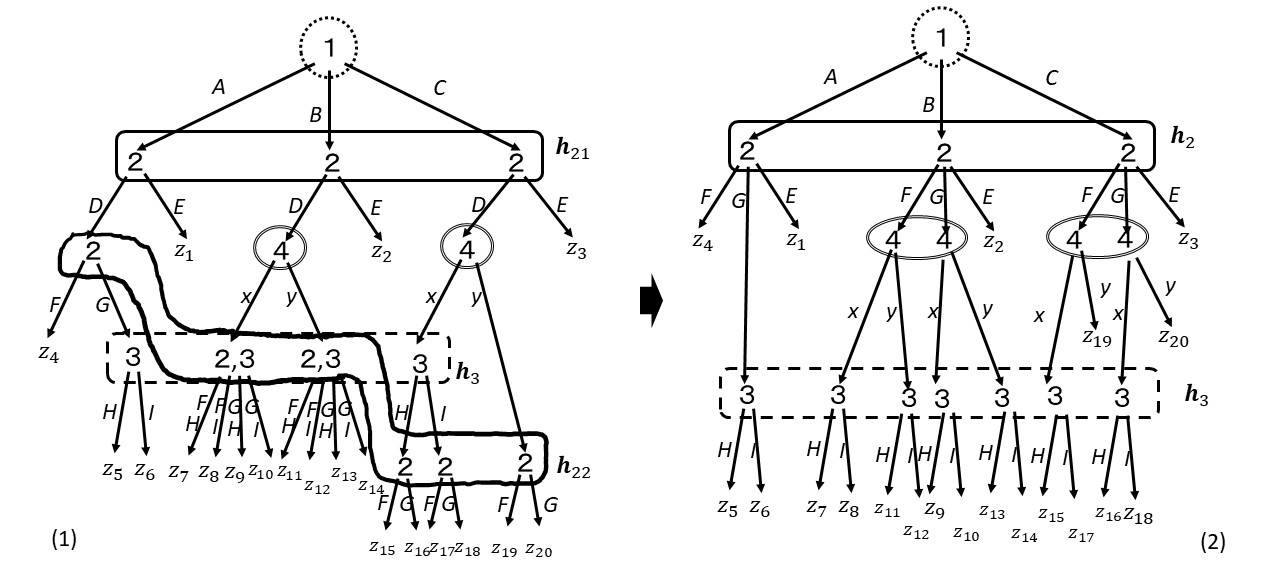}
  \caption{The inverse of Coalescing may break UO}
  \label{fig:Coaib}
\end{figure}

\subsection{Interchange/Simultanizing}

Consider $i \in I$, $h \in H$ with $i \notin I(h)$ and $\mathbf{d}_{i} \subseteq \mathbf{h}_{i}\in \mathbf{H}_{i}$ such that $h \lessdot_{i} \mathbf{d}_{i}$. We say that $G$ has an \emph{Interchange/Simultanizing (IS) opportunity} $(h,\mathbf{d}_{i})$ and define $\sigma (G; h,\mathbf{d}_{i}) =\langle \tilde{I}, \tilde{\bar{H}},(\tilde{A}_{i},\tilde{\mathbf{H}}_{i})_{i \in \tilde{I}} \rangle$ as follows:

\begin{itemize}
\item $\tilde{I} = I$ and $\tilde{A}_{j}=A_{j}$ for each $j \in I$.

\item $\tilde{\bar{H}}$ coincides with $\bar{H}$ at histories unrelated to or before $h$, the same as cases (i) and (ii) in the definition of $\bar{H}$ in Coalescing. Every history following $h$ need to be replaced in $\tilde{\bar{H}}$ due to the shifting-up of $\mathbf{d}_{i}$ to $h$. Suppose that $F_{i}(\mathbf{h}_{i}) (=F_{i}(\mathbf{d}_{i}))=\{c_{i1},...,c_{ik}\}$. For each $g \succ h$, we discuss two cases:

\begin{itemize}
\item $h \prec g \preceq h^{\prime}$ for some $h^{\prime} \in \mathbf{d}_{i}$, i.e.,  $g = ha^{1}...a^{p}$ where $a^{1} \in \prod_{j \in I(h)}F_{j}(h)$. Then $g$ is replaced by  $k$ replicas $\tilde{g}_{1}, ..., \tilde{g}_{k}$ in $\tilde{\bar{H}}$: for each $t=1,...,k$, $\tilde{g}_{t} = h(c_{it}, a^{1})a^{2}...a^{p}$;

\item $g > \mathbf{d}_{i}$, i.e., $g = ha^{1}...a^{p} (c_{it}, c_{-i})d^{1}...d^{q}$, $t \in \{1,...,k\}, p,q \in \mathbb{N}_{0}$. Then $\tilde{g} = h(c_{it}, a^{1})a^{2}...$ $a^{p}c_{-i}d^{1}...d^{q}$. Note that if at the prefix of $g$ in $\mathbf{d}_{i}$ only player $i$ is active, we directly connect the components before and after it, i.e., $\tilde{g} = h^{\prime}(c_{it}, a_{-i})b^{1}...b^{p}d^{1}...d^{q}$; especially, in this case if $g$ is an immediate successor of a history in $\mathbf{d}_{i}$, then $\tilde{g}$ coincides with some replica of its immediate predecessor.

\end{itemize}

\item Each $\mathbf{g}_{j}$ with $\mathbf{g}_{j} \neq \mathbf{h}_{i}$ has to replace each $g$ between $h$ and $\mathbf{d}_{i}$ by its replicas $\tilde{g}_{1},...,\tilde{g}_{k}$ and each $g$ following $\mathbf{h}_{i}$ by $\tilde{g}$. $\mathbf{h}_{i}$ has to replace histories in $\mathbf{d}_{i}$ by $h$.

\end{itemize}

\begin{figure}
  \includegraphics[width=\linewidth]{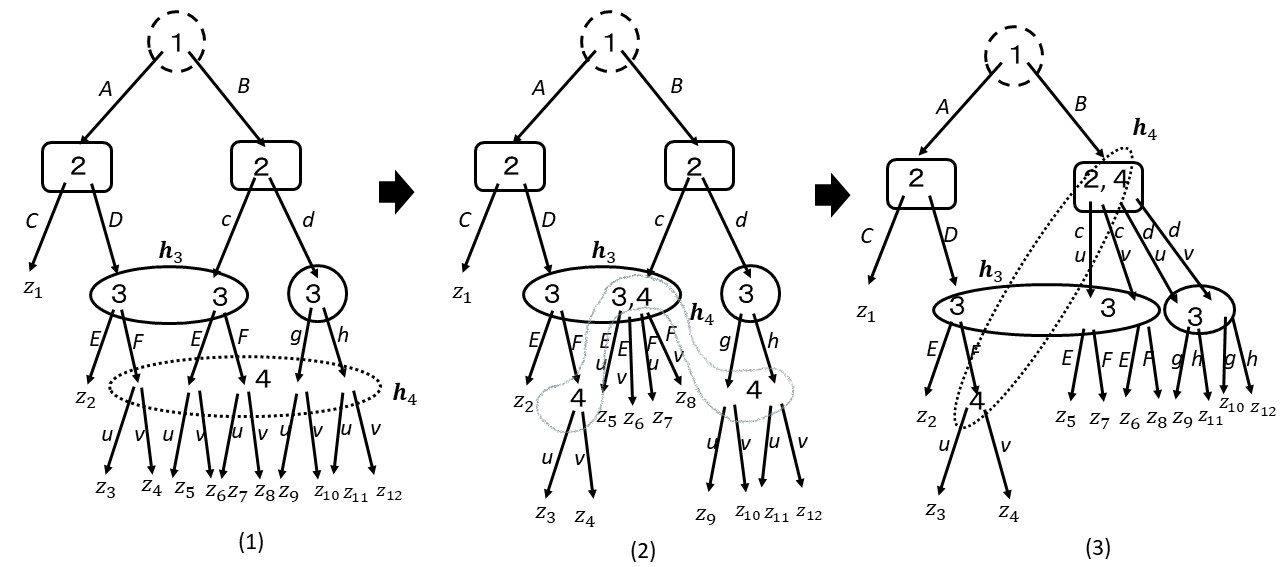}
  \caption{IS opportunities that preserves or breaks UO}
  \label{fig:ISe}
\end{figure}

Note that in $\sigma(G;h,\mathbf{d}_{i})$,  $\tilde{I}(h)=I(h)\cup \{i\}$. Also, for each $h^{\prime} \in \mathbf{d}_{i}$, if $|I(h^{\prime})|>1$, $\tilde{I}(\tilde{h^{\prime}})=I(h^{\prime}) \setminus \{i\}$; if $I(h) = \{i\}$, then $\tilde{h^{\prime}}$ coincides a with a replica of its immediate successor in $G$, and  $\tilde{I}$ and $\tilde{F}_{i}$ alter accordingly.

\smallskip
Similar to Coalescing, there is a natural correspondence, denoted by $\sigma: \bar{H} \twoheadrightarrow \tilde{\bar{H}}$, which indicates the change of each history through the transformation: for each $g \in \bar{H}$, if $g$ is before or unrelated to $h$, $\sigma (g) = \{g\}$; if $g$  is between $h$ and $\mathbf{h}_{i}$, then $\sigma (g) = \{\tilde{g}_{1},...,\tilde{g}_{k}\}$; if $g > \mathbf{h}_{i}$,  $\sigma (g) = \{\tilde{g}\}$. Also, $|\sigma (g) | = |F_{i}(\mathbf{h}_{i})|$ if $g$ is between $h$ and $\mathbf{h}_{i}$, and
$|\gamma (g) | = 1$ otherwise. In the following, when $\sigma (h)$ is a singleton, we sometimes do not differentiate $\sigma (h)$ from its element. For $g \in \mathbf{d}_{i}$, we define $\sigma (g;\mathbf{h}_{i})= \{h\}$ to emphasize that \emph{for player} $i$, histories in $\mathbf{d}_{i}$ is now synchronized with $h$.

Since no information set perishes in an IS operation, there is a bijection between $\mathbf{H}$ and  $\tilde{\mathbf{H}}$ which is denoted by $\sigma$, i.e., $\sigma (\mathbf{g}_{j}) = \mathbf{g}_{j}$ for each $\mathbf{g}_{j} \in \mathbf{H}_{j}$, $j \in I$.
\smallskip

For example, $(Bc, \{BcE, BcF\})$ is an IS opportunity in the structure in Figure \ref{fig:ISe} (1). By applying the transformation, we obtain the structure in (2) where $(B, \{Bc, Bdg,$ $Bdh\})$ is an IS opportunity. Yet the transformation on it  breaks UO and we obtain the structure in (3). The problem here is caused by $\mathbf{h}_{4}$'s partial crossing of $\mathbf{h}_{3}$. To preserve UO, we need to exclude those IS opportunities.

Formally, an IS opportunity $(h,\mathbf{d}_{i})$ with $\mathbf{d}_{i} \subseteq \mathbf{h}_{i}$ is called \emph{non-crossing} iff there is no $\mathbf{h}$ satisfying the following two conditions: (1) $\mathbf{h} < g$ or $g \in \mathbf{h}$ for some $g \in \mathbf{d}_{i}$, and (2) $\mathbf{h} > g$ for some $g \in \mathbf{h}_{i} \setminus \mathbf{d}_{i}$. In words, a non-crossing IS opportunity does not allow $\mathbf{d}_{i}$ to go over $\mathbf{h}$ while another part of $\mathbf{h}_{i}$ stays following $\mathbf{h}$. For example, the transformation from Figure \ref{fig:ISe} (2) to (3) is not non-crossing because $\mathbf{h}_{4}$ ``partially crossed'' $\mathbf{h}_{3}$. We use $\sigma_{NC}$ to denote the operator  on a non-crossing IS opportunity.

Battigalli et al. \cite{blm20} show that an IS preserves behavioral equivalence. The following lemma shows that a non-crossing IS preserves UO.

\begin{lemma}
\textbf{(A non-crossing IS preserves UO)}. Let $(h,\mathbf{d}_{i})$ be a non-crossing IS opportunity in $G$. If $G$ satisfies UO, so does $\sigma_{NC} (G; h,\mathbf{d}_{i})$.
\end{lemma}

\begin{proof}
The IS transformation breaks UO only if $\mathbf{h}_{i}$ ``partially crossed'' some information set, which is forbidden by the definition of non-crossing IS. 
\end{proof}

Similar to Coalescing, the inverse of $\sigma_{NC}$ may break UO. We restrict the range of $\sigma_{NC}^{-1}$ to the structures in $\mathcal{G}_{UO}$. Then the converse of Lemma 2 holds by definition.

\smallskip
In the following, we will use the second element in a transformation opportunity frequently. So we give it a name. In a Coalescing opportunity $(\mathbf{h}_{i}, \mathbf{h}_{i}^{\prime})$, we call $\mathbf{h}_{i}^{\prime}$ the \emph{mover}. In a (non-crossing) IS opportunity $(h, \mathbf{d}_{i})$ with $\mathbf{d}_{i} \subseteq \mathbf{h}_{i}$, we call $\mathbf{h}_{i}$ the \emph{mover} and $\mathbf{d}_{i}$ the \emph{sub-mover}.

\subsection{Characterization}

Consider an extensive game structure $G$ satisfying UO. It is called \emph{minimal} iff it does not have any Coalescing or IS opportunity; it is called \emph{minimal with respect to UO} iff there is no any Coalescing or IS opportunity in it which does not destroy UO. The latter does not imply the former. For example, the structure in Figure \ref{fig:MUD} (1) is minimal with respect to UO but has an IS opportunity $(B, \{Bc, Bd\})$ which destroys UO, as shown in Figure \ref{fig:MUD} (2).

\begin{figure}
\centering
  \includegraphics[width=0.7\linewidth]{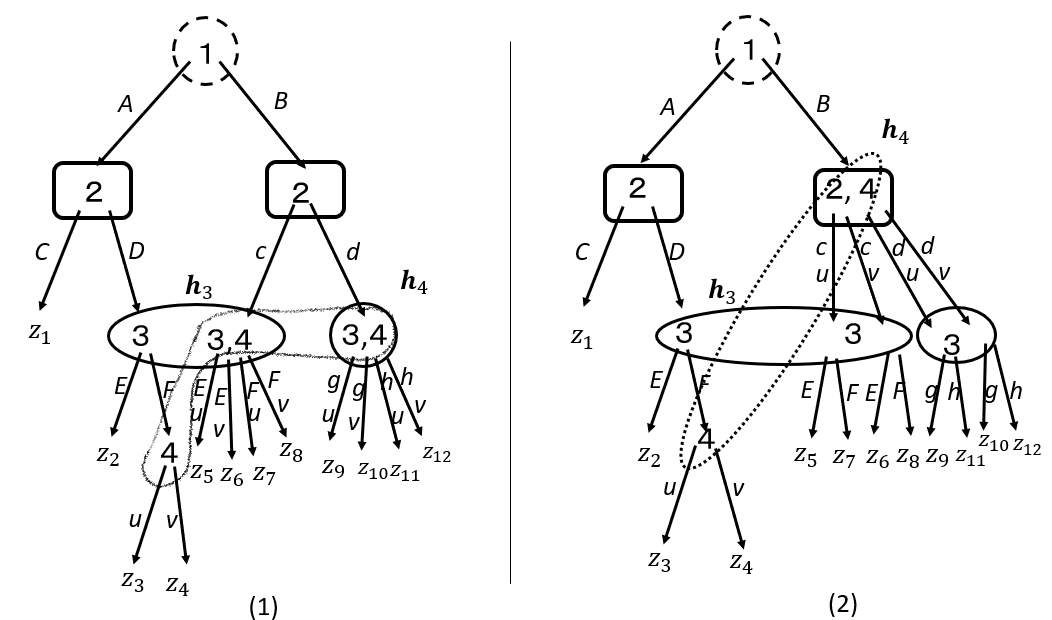}
  \caption{Minimal compactifications with and without UO}
  \label{fig:MUD}
\end{figure}

\begin{theorem}
\textbf{(Coalescing and non-crossing IS characterizes behavioral equivalence with UO)}. For extensive game structures $G$ and $G^{\prime} \in \mathcal{G}_{UO}$, the following are equivalent.

(a) $G$ and $G^{\prime}$ are behaviorally equivalent.

(b)  Both $G$ and $G^{\prime}$ can be transformed into a minimal structure with respect to UO through iteratively applying Coalescing and non-crossing IS, up to isomorphisms.

(c)  Both $G$ and $G^{\prime}$ can be transformed into a structure in $\mathcal{G}_{UO}$ through iteratively applying Coalescing and non-crossing IS, up to isomorphisms.

(d) $G$ can be transformed into $G^{\prime}$, up to isomorphisms, through a (possibly empty) finite sequence of $\gamma$ and $\sigma_{NC}$ and their inverses.

\end{theorem}

\begin{figure}
  \includegraphics[width=\linewidth]{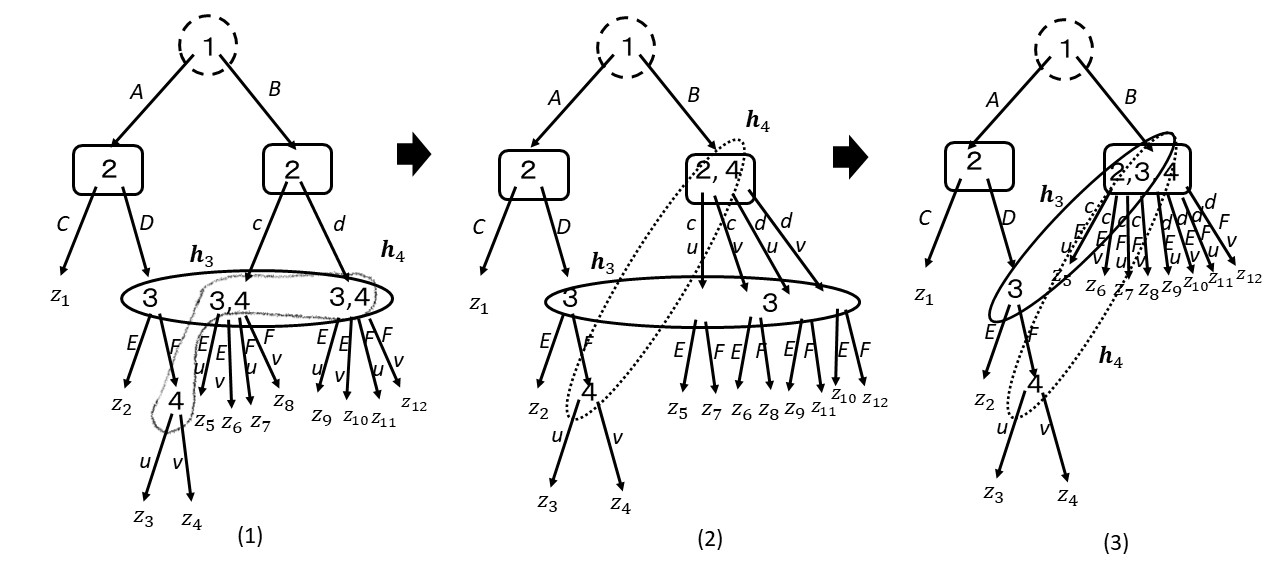}
  \caption{A compactification sequence where UO is violated \emph{en route}}
  \label{fig:UOm}
\end{figure}

Battigalli et al. \cite{blm20}'s Theorem 1 characterizes behavioral equivalence by Coalescing and IS and shows that there is a unique minimal structure for each behavioral equivalence class. We refine their result by taking UO into account. Our result shows that behavioral equivalence \emph{in} $\mathcal{G}_{UO}$ is characterized by Coalescing and non-crossing IS; also, each equivalence class in $\mathcal{G}_{UO}$ has a minimal structure.

Note that Theorem 1 does not mean that if $G, G^{\prime} \in \mathcal{G}_{UO}$ are behaviorally equivalent, then \emph{each} transformation sequence from one to the other have every step satisfying UO. For example, in the transformation sequence in Figure \ref{fig:UOm}, though the structures in (1) and (3) satisfy UO, the one in (2) does not. However, Theorem 1 means that there \emph{exists} a transformation sequence where each step satisfies UO. For the structure in Figure \ref{fig:UOm} (1), it is done by applying non-crossing IS first on $(B, \{Bc,Bd\})$ with $\{Bc,Bd\} \subseteq \mathbf{h}_{3}$ and then on $(B, \{Bc,Bd\})$ with $\{Bc,Bd\} \subseteq \mathbf{h}_{4}$.

\smallskip

The proof of Theorem 1 needs some lemmas. The first one states that the relations between information sets can neither be created or destroyed through transformations. It is called the \emph{law of conservation of relations} to emphasize its similarity with the law of conservation of mass in physics and chemistry. Note that only the \emph{existence/nonexistence of a relation} between two information sets is unchanged; the \emph{content} of the relation may change. For example, $\mathbf{h}_{4}$ follows $\mathbf{h}_{3}$ in Figure \ref{fig:ISe} (1) while the simultaneity emerges in (2) after the IS transformation

Two information sets $\mathbf{h}$ and $\mathbf{g}$ are \emph{related} iff $\mathbf{h}<\mathbf{g}$, $\mathbf{h} \sim\mathbf{g}$, or $\mathbf{h}>\mathbf{g}$ holds.

\begin{lemma}
\textbf{(The law of conservation of relations between information sets)} Let $G$ be an extensive game structure and $\iota \in \{\gamma, \sigma\}$. Two information sets $\mathbf{f}$ and $\mathbf{g}$ are related in $G$ if and only if $\iota(\mathbf{f})$ and $\iota (\mathbf{g})$ are related in $\iota (G)$.
\end{lemma}

\begin{proof}
Suppose that $\mathbf{f}$ and $\mathbf{g}$ are related in $G$. Then some $f \in \mathbf{f}$ and $g \in \mathbf{g}$ are in a chain, i.e., $f \prec g$, $f = g$, or $f \succ g$. If neither $\mathbf{f}$ or $\mathbf{g}$ is the mover of the opportunity that $\iota$ is applied on or one is the mover in an IS opportunity but the history in question is not in the sub-mover, then $\iota (f)$ and $\iota(g)$, be they singletons or sets of replicas, are still in a chain and their relative positions are unchanged (in the cases of replicas, the relative positions are preserved by some elements in the set).

Without loss of generality, suppose that $\mathbf{f}$ is the mover and $f$ is in the sub-mover if it is an IS opportunity. Since a transformation shifts histories in the mover along the chain of their respective predecessors, $\iota(f, \mathbf{f})$ and $\iota(g)$ are still in a chain. Hence the two information sets are still related after $\iota$ is applied.

The cases of the inverse transformations can be shown in similar manners. 
\end{proof}

The next lemma is first proved
in Newman \cite{ne42} (also see Apt \cite{ap11}).
It gives a sufficient condition for order-independence in an abstract reduction system. 

An \emph{abstract reduction system}
is a pair $(X  , \rightarrow )$, where $X$ is a non-empty set and $ \rightarrow $ is a binary relation on $X$. An element $x  \in X$ is called an \emph{endpoint} in $(X, \rightarrow )$ iff there is no $x ^{ \prime } \in X $ such that $x  \rightarrow x ^{ \prime }$. We say that $\{x _{n} :n =0 ,1 ,\ldots \}$ (finite or infinite) in $X$ is a $ \rightarrow $\emph{-sequence} iff $x _{n} \rightarrow x _{n +1}$ (as far as $x _{n +1}$ is defined). We use $ \rightarrow ^{ \ast }$ to denote the reflexive and transitive closure of $ \rightarrow$. We say that $(X , \rightarrow )$ is \emph{weakly confluent} iff for each $x  ,y  ,z  \in X $, if $x  \rightarrow y$ and $x  \rightarrow z$, then $y \rightarrow ^{ \ast } x ^{ \prime }$ and $z   \rightarrow ^{ \ast }x ^{ \prime}$ for some $x ^{ \prime } \in X$.

\begin{lemma}
\textbf{(Order-independence in an abstract system)}. If an abstract reduction system $(X , \rightarrow )$ satisfies the following two conditions: \textbf{(N1)} each $ \rightarrow $-sequence is finite, and \textbf{(N2)} $(X , \rightarrow )$ is weakly confluent, then for each $x  \in X $ there is a unique endpoint $x ^{ \prime } \in X $ such that $x  \rightarrow ^{ \ast }x ^{ \prime }$. 
\end{lemma}

Let $\rightarrow \in \{\gamma, \sigma_{NC}\}$.  By Lemmas 1 and 2, $(\mathcal{G}_{UO}, \rightarrow)$ is an abstract reduction system. Since each $G \in \mathcal{G}_{UO}$ is finite and both $\gamma$ and $\sigma_{NC}$ shift up information sets, each $\rightarrow$-sequence is finite. We have to show

\begin{lemma}
\textbf{(Compactification is order-independenc)} $(\mathcal{G}_{UO}, \rightarrow)$ is weakly confluent.
\end{lemma}

\begin{proof}
 Let $G, G_{1}, G_{2} \in \mathcal{G}_{UO}$ with $G \rightarrow G_{1}$ and $G \rightarrow G_{2}$. We show that $G_{1} \rightarrow G^{\ast}$ and $G_{2} \rightarrow G^{\ast}$ for some $G^{\ast} \in \mathcal{G}_{UO}$. We have to consider the following three cases.

Case 1. $G_{1} = \gamma (G; \mathbf{h}_{i}, \mathbf{h}^{\prime}_{i})$ and $G_{2}=\gamma (G; \mathbf{g}_{j}, \mathbf{g}^{\prime}_{j})$ for some $\mathbf{h}_{i}, \mathbf{h}^{\prime}_{i}, \mathbf{g}_{j}, \mathbf{g}^{\prime}_{j} \in \mathbf{H}$. There are three possibilities.

1.1. $i = j$ and $\mathbf{h}_{i} = \mathbf{g}_{i}$. Then $\mathbf{g}_{i}^{\prime} = \mathbf{h}_{i}^{\prime}$, i.e., $G_{1}=G_{2}$.

1.2. $i = j$ and $\mathbf{g}_{i} = \mathbf{h}_{i}^{\prime}$ or $\mathbf{h}_{i} = \mathbf{g}_{i}^{\prime}$. Without loss of generality, we assume that $\mathbf{g}_{i} = \mathbf{h}_{i}^{\prime}$. It implies that $\mathbf{h}_{i},\mathbf{h}_{i}^{\prime} (=\mathbf{g}_{i})$ and $\mathbf{g}_{i}^{\prime}$ are consecutive in $\mathbf{H}_{i}$. Therefore, after $\mathbf{h}_{i}^{\prime}$ being coalesced with $\mathbf{h}_{i}$, in $G_{1}$, $\mathbf{g}_{i}^{\prime}$ controls $\mathbf{h}_{i}$. Similarly, after $\mathbf{g}_{i}^{\prime}$ being coalesced into $\mathbf{h}_{i}^{\prime}$, in $G_{2}$, $(\mathbf{h}_{i}, \mathbf{h}_{i}^{\prime})$ is a Coalescing opportunity. Therefore, $\gamma(G_{1}; \mathbf{h}_{i}, \mathbf{g}^{\prime}_{i}) = \gamma(G_{2}; \mathbf{h}_{i}, \mathbf{h}^{\prime}_{i})$.

1.3. In any other case, the two opportunities do not interfere with each other; $(\mathbf{g}_{j}, \mathbf{g}^{\prime}_{j})$ is still a Coalescing opportunity in $G_{1}$ and so is $(\mathbf{h}_{i}, \mathbf{h}^{\prime}_{i})$ in $G_{2}$. Hence, $\gamma(G_{1}; \mathbf{g}_{j},$ $ \mathbf{g}^{\prime}_{j}) = \gamma(G_{2}; \mathbf{h}_{i}, \mathbf{h}^{\prime}_{i})$.

Case 2. One is obtained from $G$ through Coalescing and the other through non-crossing IS. Without loss of generality, we assume that $G_{1} = \gamma (G; \mathbf{h}_{i}, \mathbf{h}^{\prime}_{i})$ and $G_{2}=\sigma_{NC} (G; h, \mathbf{d}_{j})$ for some $\mathbf{h}_{i}, \mathbf{h}^{\prime}_{i}\in \mathbf{H}_{i}$, $\mathbf{d}_{j}\subseteq \mathbf{g}_{j} \in \mathbf{H}_{j}$, and $h \in H$. There are four possibilities.

2.1. $h \in \mathbf{h}_{i}^{\prime}$. Then $j \neq i$. If $|I(h)|>1$, $\sigma_{NC}(G_{1},h,\mathbf{d}_{j}) = \gamma(G_{2}; \mathbf{h}_{i}, \mathbf{h}^{\prime}_{i})$. If $|I(h)|=1$ and all histories in $\mathbf{d}_{j}$ are immediate successors of $h$, we have $\gamma(G_{2}; \mathbf{h}_{i}, \mathbf{h}^{\prime}_{i}) = G_{1}$. Suppose that some history in $\mathbf{d}_{j}$ is not an immediate successor of $h$ in $G$. In $G_{1}$, we can synchronize $\mathbf{d}_{j}$ with the immediate successors of $h$ in $G$ by applying non-crossing IS several times and obtain $G^{\ast}$. In $G_{2}$, $(\mathbf{h}_{i}, \mathbf{h}^{\prime}_{i})$ is still a Coalescing opportunity, and in $G_{3} :=\gamma(G_{2}; \mathbf{h}_{i}, \mathbf{h}^{\prime}_{i})$, the immediate successor of $h$ can be synchronized with the replicas of $h$ in $G_{3}$. Finally, we also obtain $G^{\ast}$.

2.2. $h$ is between $\mathbf{h}_{i}$ and $\mathbf{h}^{\prime}_{i}$ and $\mathbf{g}_{j} \neq \mathbf{h}^{\prime}_{i}$. We assume $|F_{i}(\mathbf{h}_{i})|=k$. Then in $G_{1}$, $h$ is replaced by $k$ replicas $h_{1},...,h_{k}$, and $\mathbf{d}_{j}$ in $G_{1}$ can be partitioned into $\mathbf{d}_{j}^{1},...,\mathbf{d}_{j}^{k}$ such that $h_{t} \lessdot_{j} \mathbf{d}_{j}^{t}$ for $t = 1,...,k$. Also, since $(h,\mathbf{d}_{j}$) does not destroy UO in $G$, neither does it in $G_{1}$. After iteratively applying $\sigma_{NC}$ for $k$ times, we obtain $G^{\ast}$ where $\mathbf{d}_{j}$ is synchronized with $h_{1},...,h_{k}$. On the other hand, $(\mathbf{h}_{i}, \mathbf{h}^{\prime}_{i})$ is still a Coalescing opportunity in $G_{2}$ and $\gamma(G_{2};\mathbf{h}_{i}, \mathbf{h}^{\prime}_{i})=G^{\ast}$.

2.3. $i = j$ and $\mathbf{d}_{i} \subseteq\mathbf{h}_{i}^{\prime}$. Since $G$ has perfect recall and each player has more than one actions at each history, it follows that $h$ is between $\mathbf{h}_{i}$ and $\mathbf{h}^{\prime}_{i}$. After $\mathbf{d}_{i}$ being simultanized with $h$ in $G_{2}$, still $(\mathbf{h}_{i},\mathbf{h}^{\prime}_{i})$ is a Coalescing opportunity in $G_{2}$. After shifting up $\mathbf{h}_{i}^{\prime}$ into $\mathbf{h}_{i}$ we actually obtain $G_{1}$. Therefore, $\gamma (G_{2}; \mathbf{h}_{i},\mathbf{h}^{\prime}_{i})= G_{1}$.

2.4. In any other case, the two opportunities do not interfere with each other, and we have $\sigma_{NC}(G_{1};h,\mathbf{d}_{j}) = \gamma(G_{2}; \mathbf{h}_{i}, \mathbf{h}^{\prime}_{i})$.

Case 3. $G_{1}=\sigma_{NC} (G, h, \mathbf{d}_{i})$ and $G_{2}=\sigma_{NC} (G, g, \mathbf{b}_{j})$ for some $h,g \in H$, $\mathbf{d}_{i} \subseteq \mathbf{h}_{i}$ and $\mathbf{b}_{j} \subseteq \mathbf{g}_{j}$ form some $\mathbf{h}_{i},\mathbf{g}_{j} \in \mathbf{H}$. There are two possibilities.

3.1. $g \in \mathbf{d}_{i}$ and $I(g) = \{i\}$ in $G$. Then, similar to case 2.1, if all histories in $\mathbf{b}_{j}$ are immediate successors of $g$, then $G_{1} = \sigma_{NC}(G_{2}; h, \mathbf{d}_{i})$. Otherwise both $G_{1}$ and $G_{2}$ can be transformed into some $G^{\ast}$ where the histories in $\mathbf{b}_{j}$ are synchronized with $g$'s immediate successors in $G$. The symmetric case where $h \in \mathbf{b}_{j}$ and $I(h) = \{j\}$ can be shown in the manner

3.2. In every other case the two opportunities do not interfere with each other, except that one transformation may decompose the other into several sub-opportunities, like in 2.2.

Here we have shown that $(\mathcal{G}_{UO}, \rightarrow)$ is weakly confluent.

\end{proof}

Note that Lemma 5 still holds if we replace $\sigma_{NC}$ by $\sigma$ and $\mathcal{G}_{UO}$ by $\mathcal{G}$.

\begin{lemma}
\textbf{(Uniqueness of the minimal structure with respect to UO)}. Let $G, G^{\prime} \in \mathcal{G}_{UO}$ be two minimal structures with respect to UO. Then $G$ and $ G^{\prime}$ are equal up to isomorphisms.
\end{lemma}

\begin{proof}
Since $\gamma$ does not interfere with $\sigma$ (and consequently $\sigma_{NC}$) and $\sigma_{NC}$ does not erase any information set, $G$ and $G^{\prime}$ have the same information sets. Suppose that $G \neq G^{\prime}$ (up to isomorphisms). The only problem is $\sigma_{NC}$ and there are two possibilities. 

One is that the orderings of information sets are the same but some history which is simultanized in one structure, say $G$, is not simultanized in the other, $G^{\prime}$. But this is impossible. Because (1) the IS opportunity is still in $G^{\prime}$, and (2) since the orderings of information sets are the same, this opportunity does not cause any partial-crossing in $G^{\prime}$ (otherwise it should have already caused one in $G$), there is a non-crossing IS opportunity in $G^{\prime}$, contradictory to the fact that $G^{\prime}$ is minimal.

Another possibility is that there are two information sets $\mathbf{h}$ and $\mathbf{g}$ having different relations in $G$ and $G^{\prime}$. There are two cases:

Case 1. $\mathbf{h}$ and  $\mathbf{g}$ have some relation in $G$ but are unrelated in $G^{\prime}$. Battigalli et al. \cite{blm20}'s Theorem 1 shows that $G$ and $G^{\prime}$ can be transformed to each other through a sequence of Coalescing, IS, and their inverses. And Lemma 3 implies that the (non)existence of a relation between two information sets is preserved in the process. Hence this case  is impossible.

Case 2. Some relation that $\mathbf{h}$ and  $\mathbf{g}$ have in $G$ perishes in $G^{\prime}$. Without lose of generality, suppose that $\mathbf{h} < \mathbf{g}$ in $G$ but this does not hold in $G^{\prime}$. It means that $h \prec g$ for some $h \in \mathbf{h}$ and $g \in \mathbf{g}$ in $G$ but this connection disappears in $G^{\prime}$. The only possibility is that the part containing $g$ in $\mathbf{g}$ has been shifted up in $G^{\prime}$ to some $g^{\prime}$ with $g^{\prime} \preceq h$. Yet this shifting-up is impossible in $G$ (otherwise $G$ is not minimal) even though the IS opportunity still exists. The only reason should be that some $\mathbf{f} <\mathbf{g}$ ``blocks'' $\mathbf{g}$ in $G$, i.e., if $g$ shifts up, it would partially cross $\mathbf{f}$ and mess up UO. This problem did not happen in $G^{\prime}$, which means that some information set $\mathbf{e}$ in $G$ blocks $\mathbf{f}$. Note that $\mathbf{e} \neq \mathbf{g}$ (otherwise both $\mathbf{f} >\mathbf{g}$ and $\mathbf{f} <\mathbf{g}$, which implies that $G \notin \mathcal{G}_{UO}$, a contradiction). This process can continue, and we obtain an infinite sequence of distinct information sets $\mathbf{g} > \mathbf{f} >\mathbf{e}>...$ in $G$,  which is impossible. Hence $\mathbf{h} < \mathbf{g}$ also holds in $G^{\prime}$. The case of $\mathbf{h} \sim \mathbf{g}$ can be proved in a similar manner. 

\end{proof}

\noindent \emph{Proof of Theorem 1.} (a) $\rightarrow$ (b): Battigalli et al \cite{blm20}'s Theorem 1 shows that Coalescings and ISs (hence non-crossing ISs) preserve behavioral equivalence. If $G$ and $G^{\prime}$ are behaviorally equivalent and satisfy UO, we can iteratively apply Coalescing and non-crossing IS on them and reach two minimal structures with respect to UO. Lemma 6 implies that those two are equal (up to isomorphisms). Hence (b) holds.

(b) $\rightarrow$ (c): Logical implication.

(c) $\rightarrow$ (b): Lemma 5.

(b) $\rightarrow$ (d): Let $G^{\ast}$ be the minimal structure with respect to UO obtained from $G$ and $G^{\prime}$. Then there are $\iota_{1},...,\iota_{m}, \lambda_{1},...,\lambda_{n} \in \{\gamma, \sigma_{NC}, \gamma^{-1}, \sigma_{NC}^{-1}\}$, $m,n \in \mathbb{N}_{0}$, such that  $G^{\ast} \stackrel{iso}{\in} \iota_{1}\circ ...\circ\iota_{m}(G)$ and $G^{\ast} \stackrel{iso}{\in} \lambda_{1}\circ ...\circ\lambda_{n}(G^{\prime})$. Therefore, $G^{\prime} \stackrel{iso}{\in} \lambda_{n}^{-1} \circ ...\circ \lambda_{1}^{-1}\circ \iota_{1}\circ ...\circ\iota_{m}(G)$.

(d) $\rightarrow$ (a): Battigalli et al \cite{blm20}'s Theorem 1 shows that Coalescings and ISs (hence non-crossing ISs) preserve behavioral equivalence, and Lemmas 1 and 3 and definitions of $\gamma^{-1}$ and $\sigma_{NC}^{-1}$ show that the transformation process preserves UO. $\blacksquare$

\section{Backward Dominance Procedure and Monotonicity}\label{sec:cpa}

\subsection{Backward dominance Procedure}

The definitions here follow Perea \cite{pe14}. See Chapter 8 in Perea \cite{pe12} for a detailed discussion.

Fix $G= \langle I, \bar{H}, (A_{i},\mathbf{H}_{i})_{i \in I} \rangle$ satisfying UO. We add a von Neumann-Morgenstein utility function $v_{i}: Z \rightarrow \mathbb{R}$ for each player $i \in I$ to $G$ and obtain an \emph{extensive game} $\Gamma = \langle G, (v_{i})_{i \in I} \rangle$ based on $G$. Typical extensive games are shown in Figure \ref{fig:NUL}. Recall that $\zeta:\mathcal{S} \rightarrow Z$ is the mapping that assigns to each strategy profile the corresponding terminal. The \emph{Z-reduced static game} of an extensive game $\Gamma = \langle G, (v_{i})_{i \in I} \rangle$ is a tuple rn$_{Z}(\Gamma) : = \langle I, (\mathcal{S}_{i},u_{i})_{i \in I} \rangle$ where for each player $i \in I$, the utility function $u_{i}: \mathcal{S} \rightarrow \mathbb{R}$ is defined by $u_{i} (s) = v_{i}(\zeta (s))$ for each strategy profile $s \in \mathcal{S}$.

Let $\Gamma = \langle G, (v_{i})_{i \in I} \rangle$ be an extensive game based on $G$ and rn$_{Z}(\Gamma) = \langle I, (\mathcal{S}_{i},u_{i})_{i \in I} \rangle$. Consider an information set $\mathbf{h}_{i} \in \mathbf{H}_{i}$ of some player $i \in I$. A strategy profile $s$ is said to \emph{reach} $\mathbf{h}_{i}$ iff $\zeta (s) \in Z(\mathbf{h}_{i})$; in words, the history induced by $s$ crosses $\mathbf{h}_{i}$. We use $\mathcal{S}(\mathbf{h}_{i})$ to denote the set of all strategy profiles that reach $\mathbf{h}_{i}$, and we define $\mathcal{S}_{i} (\mathbf{h}_{i}) =$ Proj$_{\mathcal{S}_{i}} \mathcal{S}(\mathbf{h}_{i})$ and $\mathcal{S}_{-i} (\mathbf{h}_{i}) =$ Proj$_{\prod_{j \neq i}\mathcal{S}_{j}} \mathcal{S}(\mathbf{h}_{i})$ as the projections of $\mathcal{S}(\mathbf{h}_{i})$ on $\mathcal{S}_{i}$ and $\mathcal{S}_{-i} (:= \prod_{j \neq i} \mathcal{S}_{j})$, respectively. 

\begin{figure}
\centering
  \includegraphics[width=0.9\linewidth]{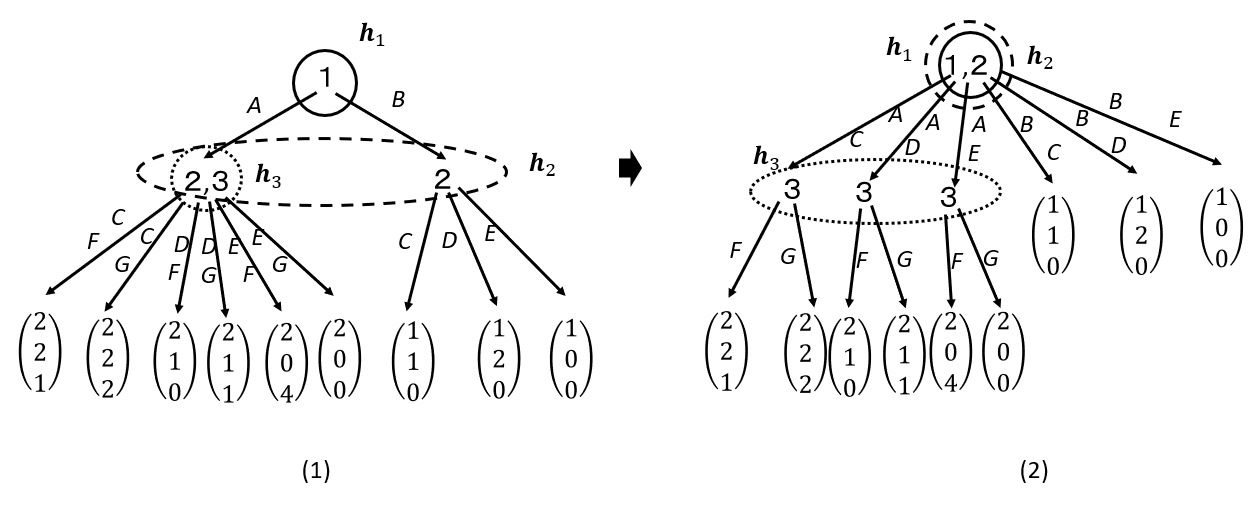}
  \caption{An invariant transformation may alter strict domination in a BD}
  \label{fig:NUL}
\end{figure}

A subset $\Psi (\mathbf{h}_{i}) \subseteq \mathcal{S}(\mathbf{h}_{i})$ is called a \emph{decision problem} iff $\Psi (\mathbf{h}_{i}) = D_{i}  \times D_{-i}$ for some $D_{i} \subseteq \mathcal{S}_{i} (\mathbf{h}_{i})$ and $D_{-i} \subseteq \mathcal{S}_{-i} (\mathbf{h}_{i})$. The intuition is that player $i$ at her information set $\mathbf{h}_{i}$ faces a problem in the decision theoretical sense: $D_{i}$ is the set of actions available and $D_{i}$ contains the states of the world which represent uncertainties to her.

Given an information set $\mathbf{h}_{i}$ of player $i$ and a decision problem $\Psi (\mathbf{h}_{i}) = D_{i}  \times D_{-i}$ at $\mathbf{h}_{i}$, a strategy $s_{i} \in D_{i}$ is \emph{strictly dominated within} $\Psi (\mathbf{h}_{i})$ iff there is a mixed strategy $\mu_{i} \in \Delta (D_{i})$ such that $u_{i}(s_{i},s_{-i}) < u_{i} (\mu_{i},s_{-i})$ for all $s_{-i} \in D_{-i}$. The set of all strictly dominated strategies within $\Psi (\mathbf{h}_{i})$ of player $i$ is denoted by $sd_{i}(\Psi (\mathbf{h}_{i}))$. We define $sd(\Psi (\mathbf{h}_{i})) =sd_{i}(\Psi (\mathbf{h}_{i})) \times D_{-i}$, i.e., the set of strategy profiles in $\Psi (\mathbf{h}_{i})$ whose $i$-th component is a strictly dominated strategy of player $i$ within $\Psi (\mathbf{h}_{i})$. 

The \emph{backward dominance procedure} (BD) is defined by induction as follows: 
\begin{enumerate}\addtocounter{enumi}{-1}
  \item  \textbf{Inductive base}. For every information set $\mathbf{h}_{i} \in \mathbf{H}$, $\Psi^{0}(\mathbf{h}_{i}) := \mathcal{S}(\mathbf{h}_{i})$.
  \item  \textbf{Inductive step}. Let $n \geq 1$. Suppose that $\Psi^{n-1}(\mathbf{h}_{i})$ has been defined for all $\mathbf{h}_{i} \in \mathbf{H}$. For each $\mathbf{h}_{i} \in \mathbf{H}$, 
  \[\Psi^{n}(\mathbf{h}_{i}) :=\Psi^{n-1}(\mathbf{h}_{i})  \setminus \bigcup_{\mathbf{g}_{j} \gtrsim \mathbf{h}_{i}}sd(\Psi^{n-1}(\mathbf{g}_{j})) \]
\end{enumerate}
For each $i \in I$, a strategy $s_{i} \in \mathcal{S}_{i}$ \emph{survives} the backward dominance procedure iff there is some $s_{-i} \in \mathcal{S}_{-i}$ such that $(s_{i}, s_{-i})\in \Psi^{n}(\emptyset)$ for all $n \in \mathbb{N}_{0}$.

\medskip

For example, the BD of the game in Figure \ref{fig:NUL} (1) is given in Figure \ref{fig:NUL1}. Here, $\mathbf{h}_{2}$ and $\mathbf{h}_{3}$ are simultaneous and both follow $\mathbf{h}_{1}$. At the first round, $B$ is strictly dominated by $A$ for player $1$ at $\mathbf{h}_{1}$ and $E$ by $C$ for player 2 at $\mathbf{h}_{2}$. Also, since $\mathbf{h}_{2}$ follows $\mathbf{h}_{1}$ and is simultaneous with $\mathbf{h}_{3}$, $E$ is eliminated at $\mathbf{h}_{1}$ and $\mathbf{h}_{3}$. Note that since $\mathbf{h}_{1}$ is before $\mathbf{h}_{2}$ and $\mathbf{h}_{3}$, $B$ cannot be eliminated at the latter. We obtain $\Psi^{1}$ in Figure \ref{fig:NUL1} (2). Now $F$ becomes strictly dominated by $G$ for player 3 at $\mathbf{h}_{3}$ and is eliminated there; also, since $\mathbf{h}_{3}$ follows $\mathbf{h}_{1}$ and is simultaneous with $\mathbf{h}_{2}$, we eliminate $F$ at $\mathbf{h}_{1}$ and $\mathbf{h}_{2}$. Here we obtain $\Psi^{2}$ in Figure \ref{fig:NUL1} (3). The BD halts here and player 1's strategy $A$, player 2's $C$ and $D$, and player 3's $G$ survive.

\begin{figure}
\centering
  \includegraphics[width=\linewidth]{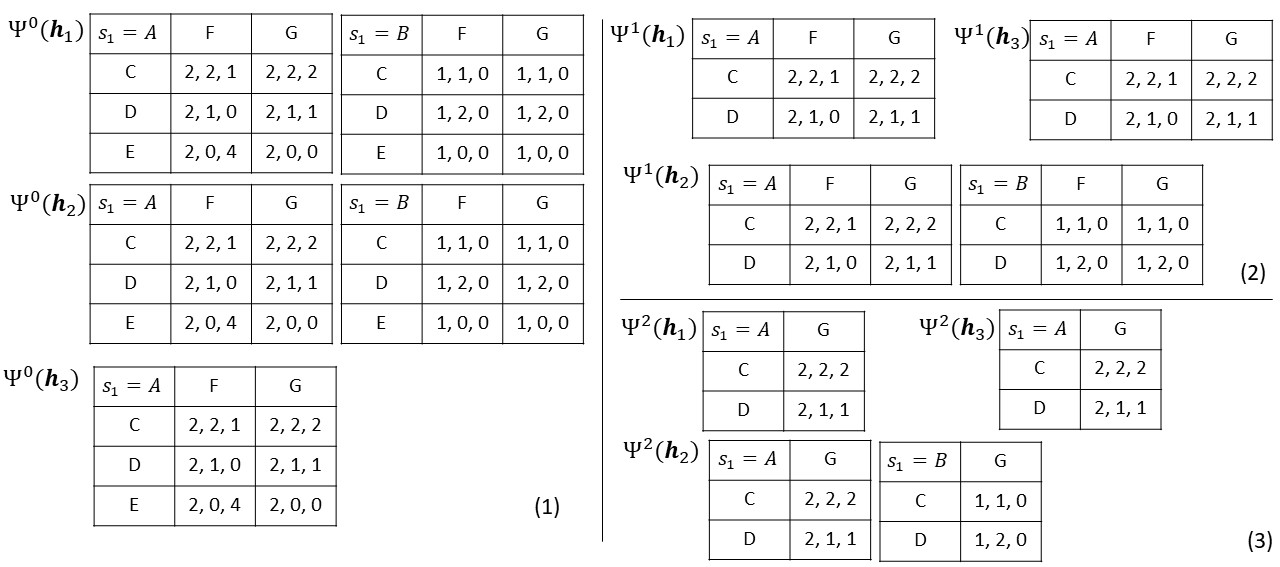}
  \caption{The backward dominance procedure of the game in Figure \ref{fig:NUL} (1)}
  \label{fig:NUL1}
\end{figure}

As noted in Section 6.2 of Perea \cite{pe14}, a transformation may alter the strategies that survive the BD. For example, the game in Figure \ref{fig:NUL} (2) is obtained from the one in Figure \ref{fig:NUL} (1) through a non-crossing IS on $(\emptyset, \mathbf{h}_{2})$. Yet Player 3's strategy $F$ survives the BD in (2), even though it does not in (1). The reason is that, due to the shifting-up of $\mathbf{h}_{2}$, player 2's strictly dominated strategy $E$ cannot be eliminated at $\mathbf{h}_{3}$ and $F$ loses its base to be strictly dominated by $G$.

\medskip
A compactification $\tau$ is called \emph{monotonic} iff every strictly dominated strategy in the original game is still strictly dominated after $\tau$ is applied. The example above suggests that a monotonic compactification has to preserve simultaneity and to weakly preserve following: for each information sets $\mathbf{h}, \mathbf{g} \in \mathbf{H}$, if  $\mathbf{h} \sim \mathbf{g}$, then  $\tau(\mathbf{h}) \sim \tau (\mathbf{g})$; if $\mathbf{h} < \mathbf{g}$, then $\tau(\mathbf{h}) \lesssim\tau (\mathbf{g})$. The next subsection will construct an monotonic compactification based on Coalescing and IS.

\subsection{Immediate compactifications}

We start with some definitions.

\begin{definition}
\textbf{(Immediate control)} . Let $i \in I$ and $\mathbf{h}_{i}, \mathbf{h}_{i}^{\prime} \in \mathbf{H}_{i}$ with $\mathbf{h}_{i} \ll_{i} \mathbf{h}_{i}^{\prime}$. We say that $\mathbf{h}_{i}^{\prime}$ \emph{immediately controls} $\mathbf{h}_{i}^{\prime}$, denoted by $\mathbf{h}_{i} \ll_{i}^{o} \mathbf{h}_{i}^{\prime}$, iff each history in  $\mathbf{h}_{i}^{\prime}$ is an immediate successor of a history in $\mathbf{h}_{i}$

\end{definition}

\begin{definition}
\textbf{(Immediate dictation)}  Let $h \in H$ and $\mathbf{d}_{i} \subseteq \mathbf{h}_{i}$ ($\mathbf{h}_{i} \in \mathbf{H}_{i}$ for some $i \in I$)such that $h \lessdot_{i} \mathbf{d}_{i}$. We say that $\mathbf{d}_{i}$ \emph{immediately dictates} $h$, denoted by $h \lessdot^{o}_{i} \mathbf{d}_{i}$,  iff each history in $\mathbf{d}_{i}$ is an immediate successor of $h$.
\end{definition}

We also call the respective opportunities in $G$ \emph{immediate}. Both definitions formulate immediacy by prohibiting any interpolation between the two components in an opportunity. Note that in Definition 1, we still allow some information set $\mathbf{h}$ to satisfy $\mathbf{h}_{i} < \mathbf{h} < \mathbf{h}_{i}^{\prime}$ provided that $\mathbf{h}$ can be partitioned by them. For example, $\mathbf{h}_{21} \ll_{2}^{o} \mathbf{h}_{22}$ holds in Figure \ref{fig:P5t} (1) even though $\mathbf{h}_{21} < \mathbf{h}_{3} < \mathbf{h}_{22}$. But $\mathbf{h}_{21} \ll_{2}^{o} \mathbf{h}_{22}$ does not hold in Figure \ref{fig:P5t} (2). The difference is that in the latter, $\mathbf{h}_{3}$ has a ``free'' history $RA$ between $\mathbf{h}_{21}$ or $\mathbf{h}_{22}$. In Figure \ref{fig:P5t} (2), $(RA, \{RAC,RAD\})$ with $\{RAC, RAD\} \subseteq \mathbf{h}_{22}$ is an immediate IS opportunity.

\begin{figure}
\centering
  \includegraphics[width=\linewidth]{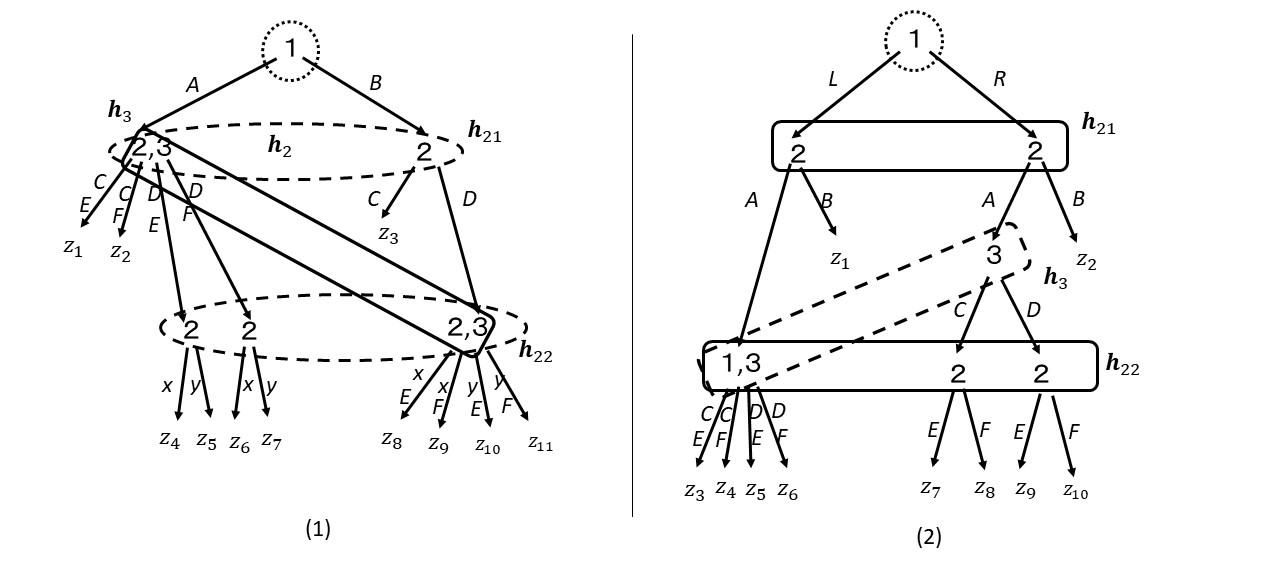}
  \caption{Illustrations of the condition in Definition 2}
  \label{fig:P5t}
\end{figure}
Note that Definition 2 does not require non-crossing. Yet this does not destroy UO in an complete immediate compactification below, as we will show in Lemma 8.

Conditions in Definitions 1 and 2 may seem too strict. Section 5.1 will discuss their necessity for preserving weak following.

\medskip
We still need some auxiliary concepts. Let $h, h^{\prime} \in H$. We say that $h$ is \emph{transitively simultaneous with} $h^{\prime}$, denoted by $h \sim^{\ast} h^{\prime}$, iff there are $h_{0},...,h_{k} \in H$, $k \geq 1$, such that (1) $h_{0} = h$ and $h_{k} = h^{\prime}$, and (2) for each $t = 0,...,k-1$, $\{h_{t},h_{t+1}\} \subseteq \mathbf{h}^{k}$ for some $\mathbf{h}^{k} \in \mathbf{H}$. Two information sets $\mathbf{h}$ and $\mathbf{h}^{\prime}$ are \emph{transitively simultaneous}, denoted by $\mathbf{h} \sim^{\ast} \mathbf{h}^{\prime}$, iff $h \sim^{\ast} h^{\prime}$ for some $h \in \mathbf{h}$ and $h^{\prime} \in \mathbf{h}^{\prime}$. Transitive simultaneity is an equivalence relation on $\mathbf{H}$.

\begin{definition}
\textbf{(Immediate compactification opportunity)} Let $\mathbf{h}_{i_{1}},...,\mathbf{h}_{i_{m}}$, $\mathbf{h}^{\prime}_{i_{1}},...,\mathbf{h}^{\prime}_{i_{m}}$, $ \mathbf{g}_{j_{1}},...,\mathbf{g}_{j_{n}}\in \mathbf{H}$, $\mathbf{d}_{j_{1}}\subseteq \mathbf{g}_{j_{1}},...,\mathbf{d}_{j_{n}} \subseteq \mathbf{g}_{j_{n}}$, and $h_{1},...,h_{n} \in H$, $m,n \geq 0$ and at least one is positive. A tuple $\Theta =\langle (\mathbf{h}_{i_{1}},\mathbf{h}_{i_{1}}^{\prime}),..., (\mathbf{h}_{i_{m}},\mathbf{h}_{i_{m}}^{\prime})$; $(h_{1},\mathbf{d}_{j_{1}}),...,(h_{n},\mathbf{d}_{j_{n}})\rangle$is called an \emph{immediate compactification opportunity} (ICO) iff the following three conditions are satisfied:

\begin{enumerate}[label=(\roman*)]
\item $\{\mathbf{h}_{i_{1}}^{\prime},... ,\mathbf{h}_{i_{m}}^{\prime},\mathbf{g}_{j_{1}},...,\mathbf{g}_{j_{n}}\}$ are transitively simultaneous,

\item $\{\mathbf{h}_{i_{1}}^{\prime},... ,\mathbf{h}_{i_{m}}^{\prime},\mathbf{d}_{j_{1}},...,\mathbf{d}_{j_{n}}\}$ are distinct,

\item $\mathbf{h}_{i_{t}} \ll_{i}^{o} \mathbf{h}_{i_{t}}^{\prime}$ for $t = 1,...,m$ and $h_{t} \lessdot^{o}_{i_{t}} \mathbf{d}_{i_{t}}$ for $t = 1,...,n$.

\end{enumerate}
\end{definition}

Recall that the word ``distinct'' is used in the sense of information sets. That is, $\mathbf{b}_{i}, \mathbf{d}_{j}$ with $\mathbf{b}_{i} \subseteq \mathbf{h}_{i} \in \mathbf{H}_{i}, \mathbf{b}_{j} \subseteq \mathbf{g}_{j} \in \mathbf{H}_{j}$  are distinct iff they contain different histories or $i \neq j$. Also, note that some $\mathbf{d}_{j_{t}}$ and $\mathbf{d}_{j_{s}}$ may belong to the the same information set $\mathbf{g}_{j}$ (hence $j_{t} = j_{s} = j$).

Given an ICO $\Theta =\langle (\mathbf{h}_{i_{1}},\mathbf{h}_{i_{1}}^{\prime}),..., (\mathbf{h}_{i_{m}},\mathbf{h}_{i_{m}}^{\prime})$; $(h_{1},\mathbf{d}_{j_{1}}),...,(h_{n},\mathbf{d}_{j_{n}})\rangle$ with $\mathbf{d}_{j_{t}} \subseteq \mathbf{g}_{j_{t}}$ for $t = 1,...,n$, we define $\mathbf{H}(\Theta) = \{\mathbf{h}_{i_{1}}^{\prime},... ,\mathbf{h}_{i_{m}}^{\prime},$ $\mathbf{d}_{j_{1}},...,\mathbf{d}_{j_{n}}\}$. An information set $\mathbf{f}_{i}$ is said to \emph{participate in} $\Theta$ iff a subset (proper or improper) of $\mathbf{f}_{i}$ is in $\mathbf{H}(\Theta)$.

\begin{definition}
\textbf{(Complete ICO)} Consider an ICO $\Theta =\langle (\mathbf{h}_{i_{1}},\mathbf{h}_{i_{1}}^{\prime}),..., (\mathbf{h}_{i_{m}},\mathbf{h}_{i_{m}}^{\prime})$; $(h_{1},\mathbf{d}_{j_{1}}),$ $...,(h_{n},$ $\mathbf{d}_{j_{n}})\rangle$ with $\mathbf{d}_{j_{t}} \subseteq \mathbf{g}_{j_{t}}$ for $t = 1,...,n$. $\Theta$ is said to be \emph{complete} iff the following two conditions are satisfied:

\begin{enumerate}[label=\alph*)]
\item Every information set that is transitively simultaneous with some information set participating in $\Theta$ also participates in $\Theta$, 

\item For each $\mathbf{f}_{i}, \mathbf{e}_{j}$ participating in $\Theta$ with $\mathbf{f}_{i} \cap \mathbf{e}_{j} \neq \emptyset$, there are some $\mathbf{b}_{i}, \mathbf{c}_{j} \in \mathbf{H}(\Theta)$ with $\mathbf{b}_{i} \subseteq \mathbf{f}_{i}$ and $\mathbf{c}_{j} \subseteq \mathbf{e}_{j}$ such that $\mathbf{f}_{i} \cap \mathbf{e}_{j} = \mathbf{b}_{i} \cap \mathbf{c}_{j}$.

\end{enumerate}
\end{definition}

Condition (b) in Definition 4 requires that the overlapping of two participants of $\Theta$ should move as a whole. We will discuss its necessity in Section 5.1.

\smallskip

Let $\Theta =\langle (\mathbf{h}_{i_{1}},\mathbf{h}_{i_{1}}^{\prime}),..., (\mathbf{h}_{i_{m}},\mathbf{h}_{i_{m}}^{\prime})$; $(h_{1},\mathbf{d}_{j_{1}}),...,(h_{n},\mathbf{d}_{j_{n}})\rangle$ be a complete ICO in $G$. We define the \emph{immediate compactification} over $G$ on $\Theta$ by
\[\tau (G, \Theta) : =\gamma (...\gamma (...\sigma (\sigma (G,h_{1},\mathbf{d}_{j_{1}}), h_{2}, \mathbf{d}_{j_{2}})...)\mathbf{h}_{i_{1}}, \mathbf{h}_{i_{1}}^{\prime})...\mathbf{h}_{i_{m}}, \mathbf{h}_{i_{m}}^{\prime})\]

\medskip 
The well-definedness of $\tau (G, \Theta)$ follows from Lemma 5..

For example, in the structure $G$ in Figure \ref{fig:ICOT} (1), $\Theta =\langle (BC, \{BCx,BCy\})$,  $(BC, $ $\{BCx,BCy\}) \rangle$, with the first $\{BCx,BCy\}$ belonging to $\mathbf{h}_{5}$ and the second to $\mathbf{h}_{4}$, is a complete ICO. The compactification is decomposed and shown in Figure \ref{fig:ICOT} (2) and (3). The structure in Figure \ref{fig:ICOT} (3) is $\tau (G, \Theta)$. 

Here are two points. First, since we do not require non-crossing, UO may be broken \emph{en route} as in (2). Nevertheless, as will be shown in Lemma 8, finally the UO will be restored. Hence, Theorem 1 implies that there is always a UO-preserving process. For the game in Figure \ref{fig:ICOT} (1), we can first synchronize $\{BCx,BCy\}$ belonging to $\mathbf{h}_{4}$ with $BC$ and then the one to $\mathbf{h}_{5}$. Second, in Figure \ref{fig:ICOT} (1), we can add $(AD, \{ADF,ADG\})$ to $\Theta$ and obtain another complete ICO. In general, given an equivalence class with respect to $\sim^{\ast}$, the complete ICO on it, if exists, may not be unique. We can define a maximal complete ICO which is a complete ICO and exhausts every compactification opportunity. But it is not necessary here.

\begin{figure}
\centering
  \includegraphics[width=\linewidth]{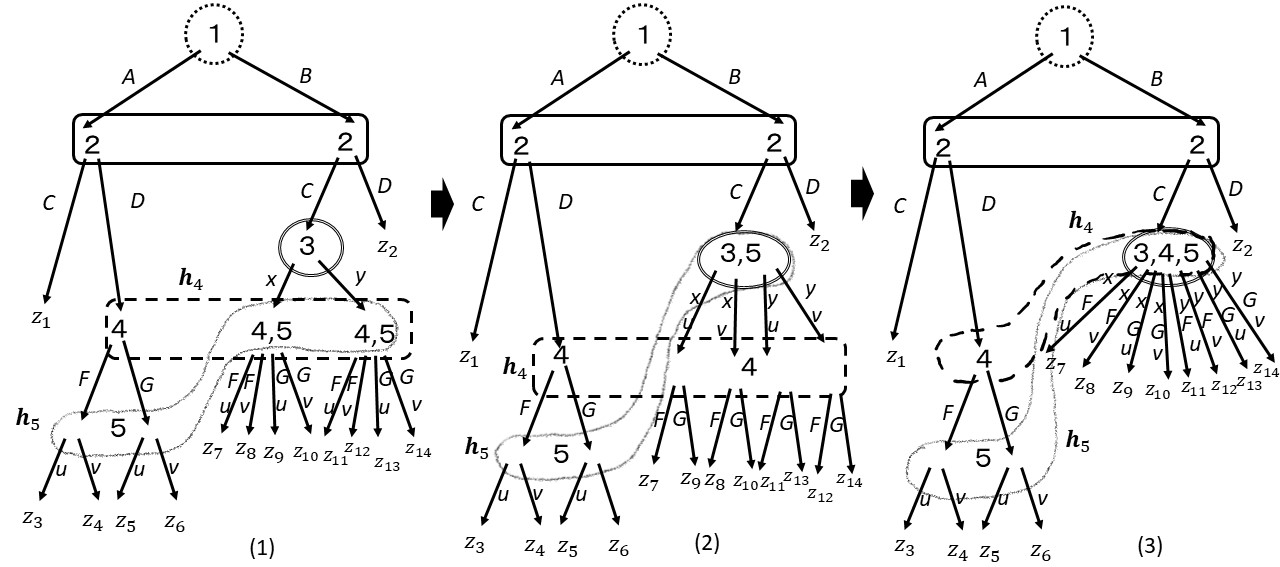}
  \caption{An ICO and the immediate compactification based on it}
  \label{fig:ICOT}
\end{figure}

\begin{lemma}
\textbf{(Preservation of simultaneity)} Let $\Theta$ a complete ICO in $G$ and $\mathbf{h}_{i}, \mathbf{g}_{j}$ be two information sets in $G$. If $\mathbf{h}_{i} \sim \mathbf{g}_{j}$ in $G$, then $\tau(\mathbf{h}_{i}) \sim \tau(\mathbf{g}_{j})$  in $\tau(G, \Theta)$.
\end{lemma}

\begin{proof}
We only need to consider the case when both $\mathbf{h}_{i}$ and  $\mathbf{g}_{j}$ participate in $\Theta$. According to their roles, we discuss the following three cases.

Case 1. Both participate in Coalescings, i.e., $(\mathbf{h}_{i}^{\prime},\mathbf{h}_{i}), (\mathbf{g}_{j}^{\prime},\mathbf{g}_{j}) \in \Theta$ for some $\mathbf{h}_{i}^{\prime} \in \mathbf{H}_{i}$ and $\mathbf{g}_{j}^{\prime} \in \mathbf{H}_{j}$. We show that $\mathbf{h}_{i}^{\prime} \sim \mathbf{g}_{j}^{\prime}$. The simultaneity implies that there is some history $h \in \mathbf{h}_{i} \cap \mathbf{g}_{j}$. By Lemma 3, $\tau (h, \mathbf{h}_{i})$ and $\tau (h, \mathbf{g}_{j})$ are also in a chain. Suppose that $\tau (h, \mathbf{h}_{i}) < \tau (h, \mathbf{g}_{j})$. It means that some history in $g_{j}^{\prime}$ is between $\mathbf{h}_{i}^{\prime}$ and $\mathbf{h}_{i}$, contradictory to Definition 1. So $\tau (h, \mathbf{h}_{i}) < \tau (h, \mathbf{g}_{j})$ is impossible. Similarly, we can show that $\tau (h, \mathbf{h}_{i}) > \tau (h, \mathbf{g}_{j})$ does not hold. Hence we have $\tau (h, \mathbf{h}_{i}) = \tau (h, \mathbf{g}_{j})$, i.e., $\mathbf{h}_{i}^{\prime} \sim \mathbf{g}_{j}^{\prime}$.

Case 2. One participates in a Coalescing and the other in an IS. Without loss of generality, assume that there are $h \in H$, $\mathbf{h}_{i}^{\prime} \in \mathbf{H}_{i}$ and $\mathbf{d}_{j} \subseteq \mathbf{g}_{j}$ such that $(\mathbf{h}_{i}^{\prime},\mathbf{h}_{i}), (h,\mathbf{d}_{j}) \in \Theta$ and there is some $g \in \mathbf{h}_{i} \cap \mathbf{d}_{j}$. Therefore, both $h$ and some history in $h^{\prime} \in \mathbf{h}_{i}^{\prime}$ are predecessors of $g$. It can be seen that $h = h^{\prime}$, since otherwise the immediacy of $(\mathbf{h}_{i}^{\prime},\mathbf{h}_{i})$ or $(h,\mathbf{d}_{j})$ would be violated. We have shown that $\tau (\mathbf{h}_{i}) \sim \tau (\mathbf{g}_{j})$.

Case 3. Both participate in ISs, i.e., $(h, \mathbf{b}_{i}), (g, \mathbf{d}_{j}) \in \Theta$ with some $h, g \in H$, $\mathbf{b}_{i} \subseteq \mathbf{h}_{i}, \mathbf{d}_{j} \subseteq \mathbf{g}_{j}$, and there is some $e \in \mathbf{b}_{i} \cap \mathbf{d}_{j}$. Since both $h$ and $g$ are predecessors of $e$, as in Case 2, we have $h = g$. 
\end{proof}

\begin{lemma}
\textbf{(Preservation of UO)}  Let $\Theta$ be a complete ICO of $G$. If $G$ satisfies UO, so does $\tau (G, \Theta)$.
\end{lemma}

\begin{proof}
Suppose that there are $\mathbf{h}_{i}, \mathbf{g}_{j}$ in $G$ such that $\tau (\mathbf{h}_{i}) < \tau (\mathbf{g}_{j})$ and $\tau (\mathbf{h}_{i}) > \tau (\mathbf{g}_{j})$ in $\tau (G, \Theta)$. Due to Lemma 3, $\mathbf{h}_{i}$ and $\mathbf{g}_{j}$ should have some relation(s) in $G$ which does (do) not violate UO. Without loss of generality, we assume that $\mathbf{h}_{i} \lesssim \mathbf{g}_{j}$, i.e., $\mathbf{h}_{i} > \mathbf{g}_{j}$ does not hold in $G$ due to UO. According to the cause of $\tau (\mathbf{h}_{i}) > \tau (\mathbf{g}_{j})$, we discuss the following cases:

Case 1. For some $h \in \mathbf{h}_{i} \cap \mathbf{g}_{j}$, $\tau (h, \mathbf{g}_{j}) < \tau(\mathbf{h}_{i})$. It follows that $\mathbf{g}_{j}$ participates in $\Theta$, and, due to the completeness of $\Theta$, $\mathbf{h}_{i} \in \Theta$. Since  the overlapping part of $\mathbf{h}_{i}$ and $\mathbf{g}_{j}$ move together, it follows from Lemma 7 that every history in the overlapping is in $\tau(\mathbf{h}_{i}) \cap \tau (\mathbf{g}_{i})$, a contradiction.

Case 2. For some history $h \in \mathbf{g}_{j}$, $\mathbf{h}_{i} < h$ and $\tau (h, \mathbf{g}_{j}) < \tau(\mathbf{h}_{i})$. It follows that $\mathbf{g}_{j}$ participates in $\Theta$. However, it implies that some history in $\mathbf{h}_{i}$ interpolates in the opportunity involving $\mathbf{g}_{j}$, violating the immediacy. 
\end{proof}

\begin{lemma}
\textbf{(Weak preservation of following)} Let $\Theta$ be a complete ICO of $G$ and $\mathbf{h}_{i}, \mathbf{g}_{j}$ be two information sets in $G$ satisfying $\mathbf{h}_{i} < \mathbf{g}_{j}$ and $\mathbf{h}_{i} \nsim \mathbf{g}_{j}$. Then $\tau(\mathbf{h}) \lesssim\tau (\mathbf{g})$, or, equivalently, $\tau(\mathbf{h}_{i})  > \tau(\mathbf{g}_{j})$ does not hold in $\tau(G, \Theta)$.
\end{lemma}

\begin{proof}
 Suppose that $\tau(\mathbf{h}_{i})  > \tau(\mathbf{g}_{j})$ holds. Since $G$ satisfies UO, $\mathbf{h}_{i} > \mathbf{g}_{j}$ does not hold in $G$. It follows that  $\mathbf{g}_{j}$ participates in $\Theta$. Yet, as in Lemma 8, it follows that some history in $\mathbf{h}_{i}$ interpolates in the opportunity involving $\mathbf{g}_{j}$, violating the immediacy. 
\end{proof}

Lemmas 7 - 9 can be summarized as follows. Consider an extensive game structure $G$ satisfying UO and a complete ICO $\Theta$  of $G$. Then $\tilde{G}:= \tau(G, \Theta)$ also satisfies UO. For each information sets $\mathbf{h}, \mathbf{g}$ in $G$, if $\mathbf{h} \sim \mathbf{g}$, then  $\tau(\mathbf{h}) \sim \tau (\mathbf{g})$; if $\mathbf{h} < \mathbf{g}$, then $\tau(\mathbf{h}) \lesssim\tau (\mathbf{g})$.

\subsection{Complete immediate compactifications and monotonicity}

Fix an extensive game structure $G= \langle I, \bar{H}, (A_{i},\mathbf{H}_{i})_{i \in I} \rangle$ satisfying UO and an extensive game $\Gamma = \langle G, (v_{i})_{i \in I} \rangle$ based on it. Theorem 1 shows that the transformations and their compositions do not alter each player's strategy sets and the mapping from strategy profiles to terminal histories. Hence in the following we use $\mathcal{S}_{i}$ ($i \in I$), $\mathcal{S}$, and $\zeta: \mathcal{S} \rightarrow Z$ on all structures behaviorally equivalent to $G$.

Here are some observations. First, given an IS (not necessarily non-crossing) opportunity $(h, \mathbf{d}_{i})$ with $\mathbf{d}_{i} \subseteq \mathbf{h}_{i}$ in $G$, for each information set $\mathbf{g}$, the strategies reaching it do not change after the Coalescing, i.e., $\mathcal{S}(\mathbf{g}) = \mathcal{S}(\sigma (\mathbf{g}))$. To see this, we only need to check $\mathbf{h}_{i}$ since its subset $\mathbf{d}_{i}$ shifts up and may be reached by more strategies. However, since $Z(h) = Z(\mathbf{d}_{i})$, every strategy that reaches $h$ also reaches $\mathbf{d}_{i}$ and \emph{vice versa}.

Second, in contrast, Coalescing enlarges the set of strategies reaching the mover. Formally, given a Coalescing opportunity $(\mathbf{h}_{i}, \mathbf{h}^{\prime})$ in $G$, we have $\mathcal{S}(\mathbf{h}_{i}^{\prime})\subsetneq \mathcal{S} (\gamma (\mathbf{h}_{i}^{\prime}))$. Indeed, since $Z(\mathbf{h}_{i}^{\prime}) = Z(\mathbf{h}_{i}a_{i}^{\ast})$ for some $a_{i}^{\ast} \in F_{i}(\mathbf{h}_{i})$ and $|F_{i}(\mathbf{h}_{i})| >1$, after the coalescing, each strategy that reaches $\mathbf{h}_{i}$ in $G$ and chooses an action there other than $a_{i}^{\ast}$ also reaches $\gamma (\mathbf{h}_{i}^{\prime}) (= \mathbf{h}_{i})$. On the other hand, note that $\mathcal{S}(\mathbf{g}) = \mathcal{S}(\gamma (\mathbf{g}))$ for every other information set $\mathbf{g}$.

The second observation concerns us since the structure of some decision problem may change due to the increase of strategies that reach some information set and the strictly dominated strategies may be influenced. Nevertheless, the following theorem states that domination are not obstructed in a complete immediate compactification.

\begin{theorem}
\textbf{(A complete ICO transformation is monotonic)} Let $\Theta$ be a complete ICO in $G$ and $\tilde{G} = \tau (G, \Theta)$. If a strategy does not survive BD in $\Gamma = \langle G, (v_{i})_{i \in I} \rangle$, neither does it in $\tilde{\Gamma} : = \langle\tilde{G}, (v_{i})_{i \in I} \rangle$; or, equivalently, if a strategy survives BD in $\tilde{\Gamma}$, so does it in $\Gamma$.
\end{theorem}

\begin{proof}
Let $\Psi^{0}, \Psi^{1},...,\Psi^{m}$ be the BD of $\Gamma$. Let $t$ be the smallest number such that some strategy eliminated from $\Psi^{t-1}$ to $\Psi^{t}$ but survives in the BD of $\tilde{\Gamma}$. Let $s_{i} \in \mathcal{S}_{i}$ and $\mathbf{h}_{i} \in \mathbf{H}_{i}$, $i \in I$, such that $s_{i} \in$ Proj$_{\mathcal{S}_{i}} \Psi^{t-1}(\mathbf{h}_{i})$, $s_{i} \notin$ Proj$_{\mathcal{S}_{i}} \Psi^{t}(\mathbf{h}_{i})$, and for some $\mu_{i} \in  \Delta($ Proj$_{\mathcal{S}_{i}} \Psi^{t-1}(\mathbf{h}_{i}))$, $u_{i}(s_{i},s_{-i}) < u_{i} (\mu_{i},s_{-i})$ for all $s_{-i} \in$ Proj$_{\mathcal{S}_{-i}}\Psi^{t-1}(\mathbf{h}_{i})$ but $s_{i}$ survives the BD in $\tilde{\Gamma}$.
We discuss three possible causes of it.

Case 1. Some $s_{i}^{\prime} \in$ supp$\mu_{i}$ is eliminated somewhere. Yet this does not cause any problem, since the mixed strategy obtained by replacing $s_{i}^{\prime}$ in $\mu_{i}$ by the (mixed) strategy dominating $s_{i}^{\prime}$ still dominates $s_{i}$.

Case 2. The ordering of information sets changed. By the choice of $t$, all strategy eliminated before round $t$ in $\Gamma$ is also eliminated in the BD of $\tilde{\Gamma}$. Hence some information set at which the elimination of a strategy provides the base for the elimination of $s_{i}$ is now shifted up and is before $\tau (\mathbf{h}_{i})$ in $\tilde{\Gamma}$. Yet by Lemmas 7 and 9, this is impossible.

Case 3. $\mathcal{S}_{-i}(\tau (\mathbf{h}_{i}))$ now contains some $s^{\ast}_{-i}$ such that $u_{i}(s_{i},s^{\ast}_{-i}) \geq u_{i}(s_{i}^{\prime}, s^{\ast}_{-i})$. According to the second observation above, it may happen when some $\mathbf{g}_{j} \in \mathbf{\Theta}$ simultaneous with $\mathbf{h}_{i}$ participates in an immediate Coalescing $(\mathbf{g}_{j}^{\prime}, \mathbf{g}_{j})$ in $\tau$. Note that $i \neq j$. We need to discuss the following possibilities.

3.1. $\mathbf{h}_{i}$ participates in an immediate coalescing $(\mathbf{h}_{i}^{\prime}, \mathbf{h}_{i})$ in $\tau$. In case 1 of the proof of Lemma 7,  we have shown that $\mathbf{h}_{i}^{\prime} \sim \mathbf{g}_{i}^{\prime}$, i.e., there is some $h \in \mathbf{h}_{i}^{\prime} \cap \mathbf{g}_{i}^{\prime}$ which is before $\mathbf{h}_{i}$. Since actions at each history is independent, it follows that each strategy that reaches $\mathbf{g}_{i}^{\prime} = \tau (\mathbf{g}_{i})$ also reaches $\mathbf{h}_{i}$. Hence the decision problem at $\mathbf{h}_{i}$ (or $\tau (\mathbf{h}_{i})$) is not enlarged by $(\mathbf{g}_{j}^{\prime}, \mathbf{g}_{j})$.

3.2. $\mathbf{h}_{i}$ participates in an immediate IS $(h, \mathbf{d}_{i})$ with $\mathbf{d}_{i} \subseteq \mathbf{h}_{i}$. Since we have shown in case 2 in the proof of Lemma 7 that $h \in \mathbf{g}_{j}^{\prime}$, and $Z(h) = Z(\mathbf{d}_{i})$, every strategy reaching $\tau (\mathbf{g}_{j}) = \mathbf{g}_{j}^{\prime}$ also reaches $\mathbf{d}_{i} \subseteq \mathbf{h}_{i}$. Hence the decision problem at $\mathbf{h}_{i}$ (or $\tau (\mathbf{h}_{i})$) is not enlarged by $(\mathbf{g}_{j}^{\prime}, \mathbf{g}_{j})$, either.

So far we have shown that $s_{i}$ should also be eliminated in the BD of $\tilde{\Gamma}$, a contradiction. Therefore, each strategy that does not survive the BD of $\Gamma$ does not survive the BD of $\tilde{\Gamma}$.  
\end{proof}

We may expect that iterated compactifications lead to a minimal game in which we can eliminate ``as many strategies as possible''. Yet here is another problem: compactification processes are not order-independent. For example, in Figure \ref{fig:G76} (1), if $\mathbf{h}_{2}$ is simultanized with $\emptyset$, we obtain the structure in (2), while if $\mathbf{h}_{3}$ is simultanized with $B$, we obtain (3). Both have no complete ICO and henceforth minimal in that sense, but they are not isomorphic and may lead to different strategies surviving  BD.

\begin{figure}
\centering
  \includegraphics[width=0.7\linewidth]{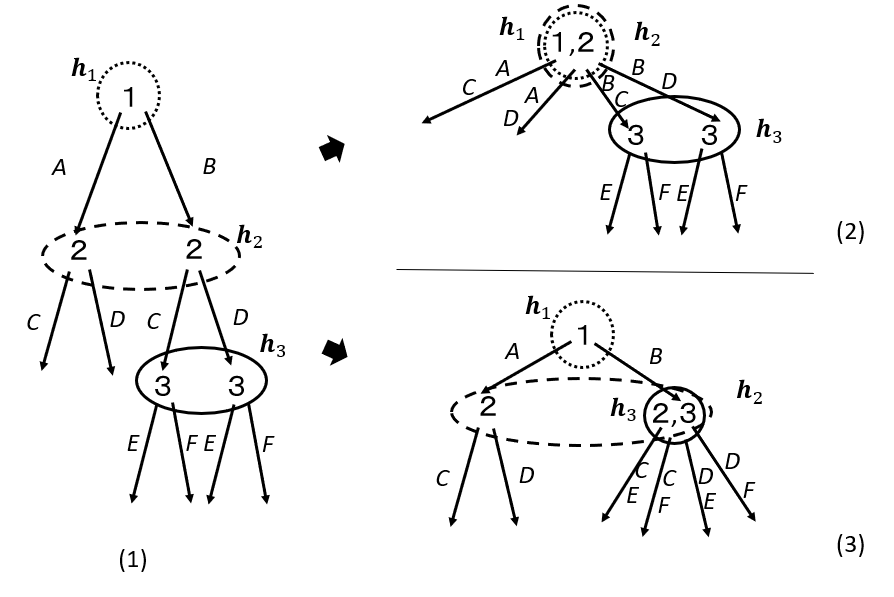}
  \caption{Equal-length absorption in a von Neumann structure}
  \label{fig:G76}
\end{figure}

Here we introduce a special complete ICO compactification called the  \emph{backward compactification},\footnote{The author thanks an anonymous reviewer for suggesting the name.} which compactifies the game ``from the leaves to the root''. To be specific, since transitive simultaneity partitions $\mathbf{H}$ into equivalence classes, we can check whether there is some complete ICO at the equivalence classes that containing immediate predecessors of terminal histories; if there is, we do the compactification, if not, we move to the classes that next to them and check the existence of complete ICO, etc.\footnote{Note that if the length of the longest history in $G$ is at least 2, then $G$ have at least two equivalence classes with respect to transitive simultaneity. Because, due to perfect recall, the root always form a singleton information set for any player active at it. An insightful discussion of this property of information sets is given in Battigalli et al \cite{blm20}, the proof of Lemma 4.} For example, Figure \ref{fig:G76} (3) is the backward compactification.

Backward compactification maximizes the number of information sets that may have some new strategy to eliminate after compactifications. Indeed, after a transformation based on some complete ICO, the information sets simultaneous to and before the participating information sets may have some new strategies to eliminate, while those following them does not. In this sense, backward compactification can be seen as a benchmark.

\section{Discussion}\label{sec:dis}

\subsection{The necessity of complete ICO for monotonicty}

Here we show by examples that the immediacy in the Definitions 1, 2, and 4 are needed for monotonicity.

Both conditions in Definitions 1 and 2 are to avoid a interpolating history ``dragging down'' some information sets and destroying simultaneity, which is illustrated in Figure \ref{fig:imb}. In (1), $(\mathbf{h}_{21}, \mathbf{h}_{22})$ is a Coalescing opportunity and $BD$ is a history between them. The transformation on $\langle (BD, \mathbf{h}_{3}), (\mathbf{h}_{21}, \mathbf{h}_{22})\rangle$ destroys the simultaneity between $\mathbf{h}_{21}$ and $\mathbf{h}_{3}$. Similarly, in Figure\ref{fig:imb} (2), transformations on $\langle (R, \mathbf{h}_{4}), (RA, \mathbf{h}_{5}) \rangle$ destroys the simultaneity between $\mathbf{h}_{4}$ and $\mathbf{h}_{5}$ due to the interpolating history $RA$ between $R$ and $\mathbf{h}_{4}$.

\begin{figure}
\centering
  \includegraphics[width=0.9\linewidth]{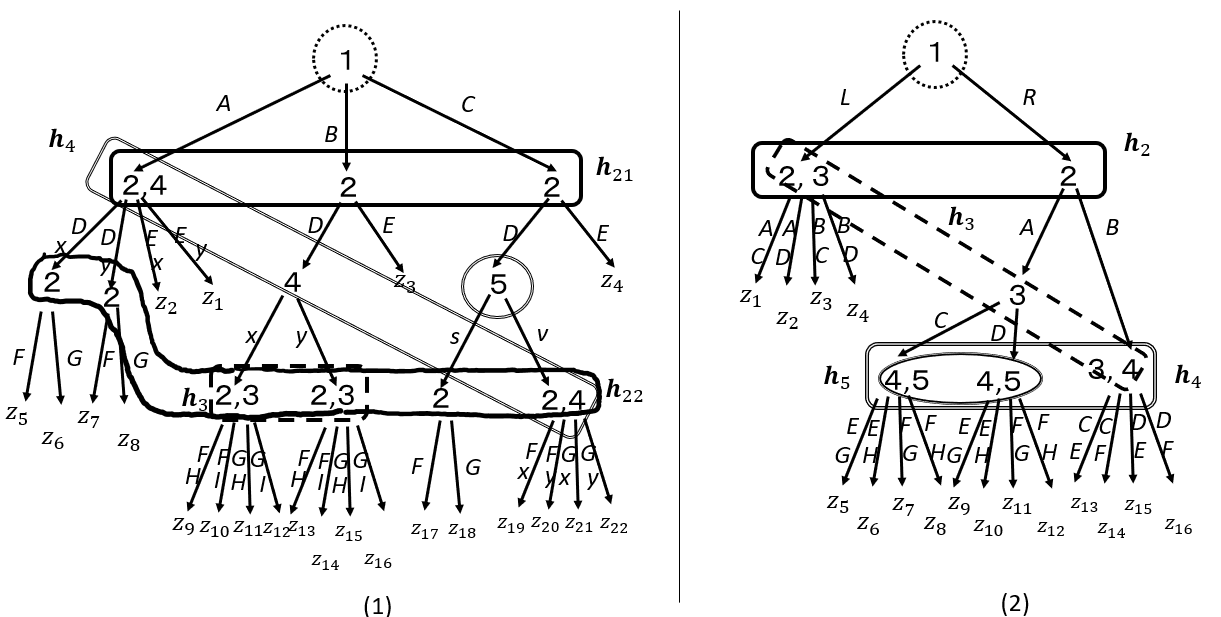}
  \caption{Simultaneity may be broken without the conditions for immediacy }
  \label{fig:imb}
\end{figure}

Condition (b) in Definition 4 is needed to preserve UO. Figure \ref{fig:ICD} (1) gives an example. The transformation on $\langle  (A, \{AC, AD\}), (BC, \{BCx, BCy\}) \rangle$ with $\{AC, AD\} \subseteq \mathbf{h}_{4}$ and $\{BCx, BCy\} \subseteq \mathbf{h}_{5}$ lead to $\mathbf{h}_{4} > \mathbf{h}_{5}$ and $\mathbf{h}_{4} < \mathbf{h}_{5}$. Here the problem is that the history in $\mathbf{h}_{4} \cap \mathbf{h}_{5}$ does not move with  them, i.e., Condition (b) in Definition 4 is violated.

Based on those examples, we may claim that the conditions are necessary for monotonicity in a weak sense. However, they are not in the strong sense. In Figure \ref{fig:ICD} (2), the IS opportunity $(RA, \{RAC, RAD\})$ with $\{RAC, RAD\} \subseteq \mathbf{h}_{22}$ seems desirable since no simultaneity is destroyed there. Yet $\langle (RA, \{RAC, RAD\})\rangle$ is not a complete ICO since $\mathbf{h}_{3}$, which is simultaneous with $\mathbf{h}_{22}$, does not participate. Similarly, there is no complete ICO in Figure \ref{fig:ICD} (2). Even though both $\mathbf{h}_{2}$ and $\mathbf{h}_{3}$ can shift up to $\emptyset$ without destroying weak following, $(\emptyset, \{BC,BD\})$ is not an immediate IS opportunity since $B$ is between them. Perhaps some  weaker conditions still satisfy monotonicity. Research in this direction is needed.
\begin{figure}
\centering
 \includegraphics[width=0.9\linewidth]{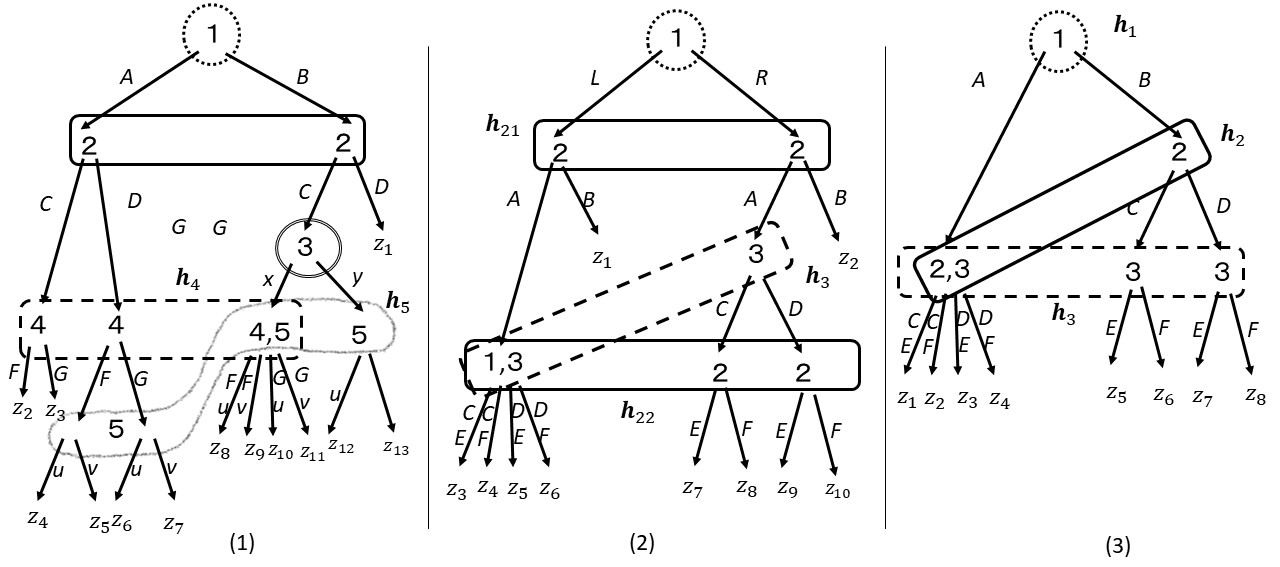}
  \caption{Immediacy are weakly but not strongly necessary for monotonicity}
  \label{fig:ICD}
\end{figure}

\subsection{Preserving the equal-length property}

As mentioned in Section \ref{sec:pre}.2, the equal-length property (EL) and von Neumann structures (vNMs) are crucial for Bonanno \cite{bo14}'s generalized backward induction. To preserve EL, we have to modify Coalescing and IS to ``collectively'' shift up information sets, which is similar to complete ICO. 

Consider a vNM $G$. For histories $h_{1}, ... , h_{n}$ $(n \geq 1)$, and an information set $\mathbf{h}_{i}$, we say that $\mathbf{h}_{i}$ \emph{completely controls} $h_{1}, ... , h_{n}$, denoted by $\{h_{1}, ... , h_{n}\}\lessdot_{i} \mathbf{h}_{i}$, iff:
\begin{enumerate}[label=\arabic*)]
\item $h_{s}$ and $h_{t}$ are not in a chain if $s \neq t$; also, $i \not \in I(h_{t})$ for $t = 1,...,n$.

\item There is a partition $\{\mathbf{d}_{i}^{1},...,\mathbf{d}_{i}^{n}\}$ of $\mathbf{h}_{i}$ such that $h_{t} \lessdot_{i} \mathbf{d}_{i}^{t}$ for $t = 1,...,n$.

\end{enumerate}
The complete control satisfies the \emph{equal-length property} (EL) iff $\ell (h_{1}) = ...=\ell (h_{n})$.

 A tuple $\Lambda = \langle (\mathbf{h}_{i_{1}}^{\prime}, \mathbf{h}_{i_{1}}),...,(\mathbf{h}_{i_{m}}^{\prime}, \mathbf{h}_{i_{m}}); (g_{1}^{1},g_{2}^{1},...,g_{L(1)}^{1}; \mathbf{g}_{j_{1}}), ..., (g_{1}^{n},...,g_{L(n)}^{n}; \mathbf{g}_{j_{n}})\rangle$ with $L: \mathbb{N} \rightarrow \mathbb{N}$ and $n,m \in \mathbb{N}_{0}$ (at least one in $m,n$ is positive). We call $\Lambda$ a \emph{synthesized opportunity} iff the following conditions are satisfied:

\begin{enumerate}[label=(\roman*)]

\item $\mathbf{h}_{i_{1}},...,\mathbf{h}_{i_{m}}$, $\mathbf{g}_{j_{1}},...,\mathbf{g}_{j_{n}}$ are distinct information sets.

\item $\mathbf{h}_{i_{t}}^{\prime} \ll_{i_{t}} \mathbf{h}_{i_{t}}$ for $t = 1,...,m$ and $\{g_{t}^{n},...,g_{L(t)}^{t}\} \lessdot_{j_{t}} \mathbf{g}_{j_{t}}$ satisfying EL for $t = 1,...,n$.

\item We use $\varphi$ to denote the composed transformation applied on $\Lambda$; for each history $h$, we define $L(h, \Lambda) = |\{g\prec h: $ each information set containing $g$ is a mover in $\Lambda\}|$. We require that for each information set $\mathbf{h} \in \mathbf{H}$ which does not participate in $\Lambda$, $L(h,\Lambda) = L(g, \Lambda)$ for each histories $h, g \in \mathbf{h}$.

\end{enumerate}

Condition (iii) means that (I) each information set in $G$ either completely follows $\mathbf{h}_{i_{1}},...,\mathbf{h}_{i_{m}}$, $\mathbf{g}_{j_{1}},...,\mathbf{g}_{j_{n}}$ or does not follow them at all, and (II) for an information set that completely following $\mathbf{h}_{i_{1}},...,\mathbf{h}_{i_{m}}$, $\mathbf{g}_{j_{1}},...,\mathbf{g}_{j_{n}}$, all histories in it move the same length upwardly.

For example, the vNM in Figure \ref{fig:GA} (1) has two synthesized opportunities,  $\Lambda_{1} = \langle (\mathbf{h}_{21}, \mathbf{h}_{22})$, $(\mathbf{h}_{31}, \mathbf{h}_{32})\rangle$ and $\Lambda_{2} = \langle (\mathbf{h}_{21}, \mathbf{h}_{22}), (\mathbf{h}_{31}, \mathbf{h}_{32}), (B, \mathbf{h}_{31}) \rangle$. Figure \ref{fig:GA} (2) is the structure obtained by applying the transformation on $\Lambda_{2}$. Note that if we want to move $\mathbf{h}_{22}$, we have to move $\mathbf{h}_{32}$, and \emph{vice versa}. Otherwise the EL of $\mathbf{h}_{5}$ would be destroyed. If we want to move $\mathbf{h}_{31}$, we have to move $\mathbf{h}_{22}$, otherwise the EL of $\mathbf{h}_{22}$ would be destroyed. But the move of $\mathbf{h}_{22}$ does not need that of $\mathbf{h}_{31}$.

The inverses of $\varphi$ should also be restricted on vNMs. A result parallel to Theorem 1 can be proved. Namely, two vNMs are behaviorally equivalent if and only if they can be transformed into each other, up to isomorphisms, through (possibly empty) finite sequences of $\varphi$ and its inverse. The proof is similar to that of Theorem 1 and is omitted here.

\medskip
A complete ICO  in a vNM is a row of transitively simultaneous information sets which can be completely shifted up to their immediate predecessors. It is straightforward to see that the monotonicity of Bonanno \cite{bo14}'s generalized backward induction holds in the transformation on a complete ICO.

 \begin{figure}
\centering
  \includegraphics[width=0.9\columnwidth]{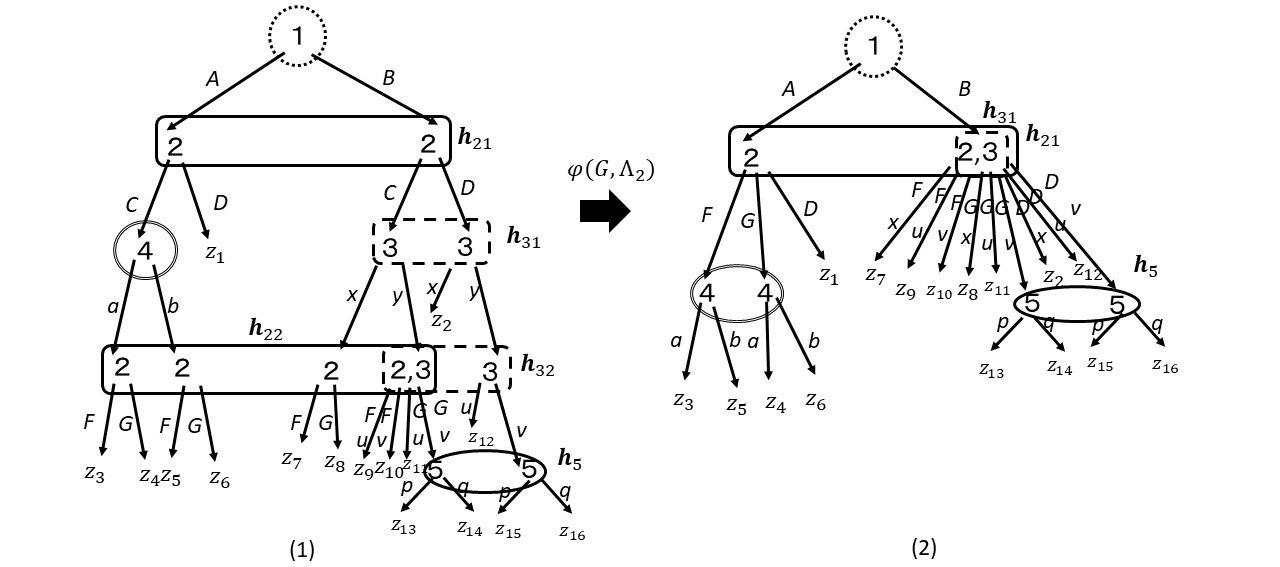}
  \caption{A synthesized opportunity and the transformation applied on it}
  \label{fig:GA}
\end{figure}

\end{document}